\documentclass[11pt]{article}

\usepackage[T1]{fontenc}
\usepackage[utf8]{inputenc}
\usepackage[hyphens]{url}
\usepackage{amsmath,amssymb,amsfonts,amsthm,bm,bbm,mathtools}
\usepackage{tikz}
\usetikzlibrary{arrows.meta,positioning,calc}
\usepackage{graphicx,microtype,booktabs,multirow}
\usepackage{array}
\usepackage{libertine}
\usepackage{xcolor}
\usepackage{enumitem}
\usepackage{setspace}
\usepackage[hang,flushmargin]{footmisc}
\usepackage[round,longnamesfirst,authoryear]{natbib}
\usepackage[
  colorlinks=true,
  hyperfootnotes=false,
  linkcolor=blue,
  citecolor=blue,
  urlcolor=blue,
  pdftitle={Ranking Experiments under Sequential Sampling},
  pdfauthor={Zihao Li and Tianhao Liu}
]{hyperref}
\usepackage[nameinlink,capitalize]{cleveref}

\usepackage[paperwidth=8.5in,paperheight=11in]{geometry}
\makeatletter
\usepackage{fancyhdr}
\makeatother

\bibpunct{(}{)}{;}{a}{,}{,}

\theoremstyle{plain}
\newtheorem{theorem}{Theorem}
\newtheorem{lemma}{Lemma}
\newtheorem{proposition}{Proposition}
\newtheorem{corollary}{Corollary}
\theoremstyle{definition}
\newtheorem{definition}{Definition}
\newtheorem{assumption}{Assumption}
\newtheorem{example}{Example}
\theoremstyle{remark}
\newtheorem{remark}{Remark}
\newtheorem*{remark*}{Remark}

\definecolor{darkblue}{RGB}{60,70,90}
\definecolor{Fcol}{HTML}{00869C}
\definecolor{Gcol}{HTML}{C9621A}
\definecolor{muted}{RGB}{103,111,124}

\crefname{assumption}{assumption}{assumptions}
\Crefname{assumption}{Assumption}{Assumptions}
\crefname{example}{example}{examples}
\Crefname{example}{Example}{Examples}

\newcommand{\E}{\mathbb E}
\newcommand{\Prob}{\mathbb P}
\newcommand{\R}{\mathbb R}

\newcommand{\one}{\mathbf 1}
\newcommand{\Law}{\mathcal L}

\newcommand{\KL}[2]{\mathrm D\!\left(#1\middle\|#2\right)}

\newcommand{\supp}{\operatorname{supp}}
\newcommand{\argmax}{\operatorname*{arg\,max}}

\title{\textbf{Ranking Experiments under Sequential Sampling}}
\author{%
  Zihao Li%
  \thanks{%
    Department of Economics, Columbia University;
    \href{mailto:zl3366@columbia.edu}{zl3366@columbia.edu}.
  }%
  \and
  Tianhao Liu%
  \thanks{%
    Cowles Foundation for Research in Economics, Yale University;
    \href{mailto:tianhao.liu@yale.edu}{tianhao.liu@yale.edu}.
  }%
}

\date{November 3, 2026}

\begin{document}

\maketitle

\begingroup
\renewcommand{\thefootnote}{}
\footnotetext{%
  We are grateful to Yeon-Koo Che, Mira Frick, Ying Gao, Yuhta Ishii, Navin Kartik, Xiao Lin, Qingmin Liu, Xiaosheng Mu, Wolfgang Pesendorfer, Philipp Strack, Kai Hao Yang, and Renjie Zhong for helpful comments.
}
\endgroup

\begin{abstract}
We rank experiments when a decision maker samples sequentially and chooses when to stop. We define two orders: \emph{large-budget dominance} compares values in every decision problem once the expected-sample budget exceeds a common threshold; \emph{small-cost dominance} compares values net of sampling costs when the cost per observation is sufficiently small in each problem. Generically, the orders coincide and are characterized by directed Kullback--Leibler (KL) dominance: larger KL divergence for every ordered pair of states. 
We also characterize the cost per copy of producing a specified experiment at scale by sampling a fixed source experiment. With a full-support prior, these cost functions are precisely sums of statewise maxima of positively weighted directed KL divergences. 
Both results rest on an exact reproduction theorem: any experiment generated by sequentially sampling one source can be reproduced exactly by sampling another, with sharp first-order KL conversion rates and a square-root remainder whose coefficient depends only on the two sources.
\end{abstract}

\section{Introduction}\label{sec:introduction}

Since \citet{Blackwell1951,Blackwell1953}, the comparison of experiments has been a central problem in information economics. Subsequent work has shown that the ranking of experiments can depend on how information is acquired.\footnote{For instance, large-sample comparisons can rank experiments that are incomparable under the Blackwell order \citep{MoscariniSmith2002,Azrieli2014,MuPomattoStrackTamuz2021}.}
In this paper, we focus on information-acquisition technologies with two features: first, the decision maker collects observations sequentially and chooses when to stop before taking an action; second, the decision maker collects a large number of observations on average. 


Such technologies are common in data-driven settings. An online platform can extend an A/B test to additional users when the initial results are inconclusive. A developer of machine-learning models can buy additional labels for a training example when existing labels disagree.\footnote{Endogenous choices of sample size have been studied in A/B testing \citep[e.g.,][]{JohariKoomenPekelisWalsh2022} and in crowdsourced labeling \citep[e.g.,][]{ShengProvostIpeirotis2008}. Moreover, the marginal costs of sampling have been documented to be low in both settings \citep{KohaviEtAl2020Trials,SnowOConnorJurafskyNg2008}.}
This paper asks how experiments should be ranked under such a technology. Specifically, all else equal, under what conditions is one experiment at least as valuable as another in every decision problem when decision makers can sample extensively?



To answer this question, we introduce two orders on experiments based on their value in decision problems. Under the first, the decision maker faces a budget on the expected number of observations she can sample; under the second, she pays a cost per observation without a budget constraint. We say that experiment $F$ \emph{large-budget dominates} experiment $G$ if, above a common budget threshold, sampling from $F$ yields at least as much value as sampling from $G$ in every decision problem. 
Equivalently, large-budget dominance asks which experiment is more valuable as the cost per observation approaches zero, holding fixed an ex ante budget on sampling expenditures.


Our second order replaces the budget constraint with a cost per observation. We say that experiment $F$ \emph{small-cost dominates} experiment $G$ if, in every decision problem, sampling from $F$ yields at least as much value, net of sampling costs, as sampling from $G$ once the common cost per observation is sufficiently small. 
Small-cost dominance thus ranks experiments when observations are cheap, and the decision maker faces no budget constraint. Unlike the common budget threshold in large-budget dominance, the cost threshold in small-cost dominance may depend on the decision problem, since whether a given sampling cost is small depends on the stakes involved.\footnote{Online Appendix~\ref{app:alternative-orders} shows that requiring a common small-cost threshold across decision problems with comparable stakes yields the same generic ranking.}


Our first main result, Theorem~\ref{thm:generic-equivalence}, shows that, generically, the two orders coincide and admit a simple characterization. For states $i\ne j$, let $I_{ij}(E)$ denote the Kullback--Leibler (KL) divergence of the signal distribution in state $i$ from the signal distribution in state $j$ under experiment $E$. 
Then $F$ dominates $G$ in either order if and only if $I_{ij}(F)>I_{ij}(G)$ for every ordered pair of distinct states.\footnote{Genericity rules out ties: $I_{ij}(F)\ne I_{ij}(G)$ for every ordered pair of distinct states.}
The theorem thus reduces solving sequential-sampling across decision problems to a static criterion: comparing the experiments' KL divergences for each ordered pair of distinct states.

This characterization contrasts with existing rankings of experiments in large samples. \citet{MoscariniSmith2002}, henceforth MS, compare experiments as the sample size grows in a fixed decision problem. \citet{MuPomattoStrackTamuz2021}, henceforth MPST, ask when sufficiently many observations from one experiment Blackwell dominate the same number of observations from another. With two states, the MS ranking is governed by Chernoff information, whereas the MPST ranking is characterized by R\'enyi divergences under suitable genericity conditions. Neither fixed-sample criterion is necessary for our dominance notions. The key difference is that both MS and MPST fix the sample size before sampling begins, whereas sequential sampling lets the decision maker choose when to stop based on the evidence accumulated.

With a fixed sample size, the decision maker collects the same number of observations whether the evidence becomes convincing early or remains inconclusive. In the latter case, she may have to act on insufficient evidence, so fixed-sample rankings depend on how quickly the probability of error declines as the sample size grows. With endogenous stopping, however, she can stop when the evidence is convincing and continue sampling otherwise, allocating more observations to histories in which additional evidence is needed. Inconclusive histories then lead to further sampling rather than necessarily to errors. The comparison thus shifts from the probability of error at a fixed sample size to the expected number of observations needed to accumulate sufficient evidence. This flexibility in allocating observations across histories can reverse the ranking of experiments, as \Cref{sec:orders} illustrates.

To see why KL dominance is necessary when stopping is endogenous, consider a decision maker who samples sequentially to distinguish between two states, $i$ and $j$. To make the probability of error small, she must accumulate sufficient evidence before stopping. In state $i$, each observation from experiment $E$ provides, on average, $I_{ij}(E)$ units of log-likelihood evidence in favor of $i$ against $j$. Thus, to first order, the expected number of observations required to achieve a given small error probability is inversely proportional to $I_{ij}(E)$. A lower evidence rate therefore requires more observations in state $i$ to achieve the same accuracy, either using more of the budget or incurring larger sampling costs. Dominance across all decision problems thus requires higher directed KL divergences for every ordered pair of states, under both the budget and the cost formulation.

We next show that KL dominance implies large-budget dominance, and hence small-cost dominance as well.\footnote{The latter implication follows from a standard budget--price duality argument.}
For a prior $\mu$, let $\mathcal P_n^\mu(E)$ denote the set of experiments generated by sequential sampling from $E$ within a budget $n$ under $\mu$, together with their garblings. Then it suffices to show \emph{production dominance}: $\mathcal P_n^\mu(G)\subseteq\mathcal P_n^\mu(F)$ for every prior $\mu$, when $n$ is sufficiently large.\footnote{In fact, \Cref{prop:inclusion} shows that production dominance is equivalent to large-budget dominance.} 
Establishing production dominance presents two challenges. First, although any given experiment in $\mathcal P_n^\mu(G)$ can be reproduced with a sufficiently large budget under $F$, the budget must work uniformly across all experiments in $\mathcal P_n^\mu(G)$. Second, the uniform budget must be eventually no greater than $n$, even as $\mathcal P_n^\mu(G)$ expands with $n$. It is not immediate why KL dominance, which compares expected log-likelihood ratios, should imply such a uniform guarantee of exact reproduction.

Our main technical result, \Cref{thm:conversion}, resolves these challenges. We define a \emph{stopping policy} as a rule that specifies when to stop sampling and what signal to report based on the observations collected, possibly using randomization. For any stopping policy under $G$ with stopping time $T$, the theorem constructs a stopping policy under $F$ with stopping time $N$ that reproduces exactly the same experiment and satisfies, in every state $i$,
\[
    \E_i^F N \leq \rho_i\,\E_i^G T + C_{F,G}\sqrt{\E_i^G T},
    \qquad
    \rho_i = \max_{j\ne i}\frac{I_{ij}(G)}{I_{ij}(F)}.
\]
The coefficient $\rho_i$ measures how many observations from $F$ are needed to match the expected log-likelihood evidence from one observation of $G$ against every alternative state. The remainder is of square-root order and, crucially, its constant $C_{F,G}$ depends only on $F$ and $G$, not on the stopping policy.

This uniformity resolves the first challenge. For any prior $\mu$, every experiment in $\mathcal P_n^\mu(G)$ can be reproduced using $F$ within the common budget $n'=\rho n+C_{F,G}\sqrt n$, where $\rho=\max_i\rho_i$. This bound follows by averaging the statewise bounds under $\mu$ and applying Jensen's inequality. Moreover, the first-order rate $\rho_i$ resolves the second challenge: KL dominance implies $\rho<1$, so $n'\leq n$ for all sufficiently large $n$, even as $\mathcal P_n^\mu(G)$ expands with $n$. 
Therefore, Theorem~\Cref{thm:conversion} implies that $\mathcal P_n^\mu(G)\subseteq\mathcal P_n^\mu(F)$ for every prior $\mu$ when $n$ is sufficiently large, and thus large-budget dominance.

Next, we apply \Cref{thm:conversion} to study a complementary question: holding the source experiment fixed, what is the cost of producing information at scale through sequential sampling? Consider a producer who draws observations from a source experiment $E$ at a fixed cost per observation and must deliver many conditionally independent copies of a prescribed experiment $H$.\footnote{This setting is particularly relevant when information is inexpensive, so a buyer may demand large batches of information.} 
She chooses a stopping policy to minimize her expected sampling cost and may use the observations jointly to produce the copies. Given a prior $\mu$, we define the \emph{large-scale unit cost} of $H$ as the limit of the minimum expected sampling cost per copy as the number of copies grows. As $H$ varies, these unit costs define a cost function over experiments induced by sequential sampling from $E$.

Our third main result, \Cref{thm:costs}, derives the large-scale unit cost and characterizes exactly which cost functions can arise from sequential sampling. Given a prior $\mu$ and a cost $c$ per observation, the large-scale unit cost of producing $H$ from $E$ is
\[
    C_\infty^E(H)
    =c\sum_i\mu_i\max_{j\ne i}
    \frac{I_{ij}(H)}{I_{ij}(E)}.
\]
To understand this formula, consider producing $m$ conditionally independent copies of $H$. In state $i$, this batch contains $mI_{ij}(H)$ units of expected log-likelihood evidence against state $j$, while each observation from $E$ supplies $I_{ij}(E)$ units on average. Thus, any stopping policy that produces the batch must use at least $mI_{ij}(H)/I_{ij}(E)$ observations on average in state $i$. Since this requirement holds for every alternative state $j$, the lower bound is determined by the largest of these ratios.
\Cref{thm:conversion} shows that this lower bound can be attained: a single stopping policy can reproduce the entire batch exactly, using at most $O(\sqrt{m})$ additional observations in expectation in each state. Since this additional cost vanishes per copy as $m\to\infty$, we obtain the large-scale unit cost formula above.


The cost functions we derive take the form $\sum_i\max_{j\ne i}\pi_{ij}I_{ij}(H)$, with positive weights $\pi_{ij}$. Conversely, we show that every function of this form can be generated by sequential sampling from some source experiment and prior. 
With two states, the cost reduces to a weighted sum of KL divergences, as characterized by \citet{PomattoStrackTamuz2023} and \citet{BloedelZhong2025} and also derived from a two-state diffusion by \citet{MorrisStrack2019}. 
With three or more states, the cost can exhibit economies of scope: producing two experiments jointly, while keeping their signals conditionally independent, may be strictly cheaper than producing them separately. This is because, for a given state, the alternative against which the evidence requirement demands the most observations may differ between the two experiments.
Producing each experiment separately may therefore generate excess evidence against other alternatives. Joint production can use this excess evidence to meet the other experiment's requirements, reducing the total number of observations.
Theorem~\ref{thm:costs} thereby provides a sequential-sampling foundation for a subclass of the subadditive Max-KL costs characterized by \citet{BordoliIijima2025}.

\subsection{Related literature}\label{sec:literature}

Our paper contributes to the literature on comparing experiments \citep{Blackwell1951,Blackwell1953}. For large sample sizes fixed in advance, the comparisons of \citet{MoscariniSmith2002} and \citet{MuPomattoStrackTamuz2021}, discussed above, are the main benchmarks. \citet{Azrieli2014} shows that requiring large samples of one experiment to Blackwell dominate equally large samples of the other is strictly more demanding than the order of \citet{MoscariniSmith2002}. \citet{FarooqFritzHaapasaloTomamichel2024} extend the analysis of \citet{MuPomattoStrackTamuz2021} to more than two states. 
Our paper differs from this strand by allowing the number of observations to depend on the realized history.

A second, conceptually related strand extends Blackwell's comparison to dynamic information structures, where information arrives over time and the decision maker can choose an action each period \citep{Greenshtein1996,deOliveira2018,WhitmeyerWilliams2024,RenouVenel2025}. We instead consider repeated independent sampling from a fixed experiment, with endogenous stopping and a single terminal decision.
Like us, \citet{KimObara2026} study stopping problems with repeated observations, but with an evolving hidden state, costless sampling, and a deadline. Their ranking at long horizons depends only on which posterior beliefs each experiment can induce, whereas ours depends on KL divergences.



Our analysis draws on classical sequential analysis \citep{Wald1945,Wald1947,WaldWolfowitz1948,ArrowBlackwellGirshick1949}. For two actions, \citet{Chernoff1959} derives a first-order $c\log(1/c)$ expansion for the loss relative to full information, and \citet{KieferSacks1963} extend the analysis to multiple actions. Our necessity result builds on the two-action expansion. These results, however, do not establish sufficiency for large-budget dominance: they are derived for fixed decision problems, whereas large-budget dominance requires a common budget threshold across all decision problems. 
A related literature on multihypothesis sequential testing characterizes KL-based sampling requirements for meeting prescribed error or risk constraints and constructs asymptotically efficient tests \citep{Lorden1977,DragalinTartakovskyVeeravalli1999}. \Cref{thm:conversion} instead establishes a KL-based bound on the sampling cost of exactly reproducing an arbitrary terminal experiment generated by another source.

Exact reproduction is conceptually related to exact simulation. In the Bernoulli factory problem, the source is a coin whose unknown bias $p$ can take a continuum of values, and the target is a coin with bias $f(p)$ for a given function $f$. The literature asks which $f$ can be simulated and how many tosses the simulation requires \citep{KeaneOBrien1994,NacuPeres2005,Huber2016,Mendo2019}. \citet{MorinaLatuszynskiNayarWendland2022} extend the problem from coins to dice. With two states and deterministic sample sizes, \citet{WangWilde2019} characterize asymptotic conversion rates by ratios of KL divergences, allowing vanishing approximation error in one state while requiring exact reproduction in the other.\footnote{Their results concern pairs of quantum states, which include pairs of probability distributions as a special case.} Like a Bernoulli factory, \Cref{thm:conversion} reproduces the target exactly using a random number of observations, and, as in \citet{WangWilde2019}, this number is governed by ratios of KL divergences. Unlike both, our result applies to general experiments over any finite number of states and to any terminal experiment generated by a stopping policy, with sharp KL conversion rates and a square-root remainder uniform over stopping policies.


The costs of producing information under sequential sampling connect our analysis to axiomatic work on information costs. With two states, these costs are weighted sums of KL divergences, the form characterized by \citet{PomattoStrackTamuz2023}; with more states, they form a subclass of the Max-KL costs of \citet{BordoliIijima2025}. Several papers derive costs of information from sequential acquisition: \citet{MorrisStrack2019} use a two-state diffusion, while \citet{DentiMarinacciRustichini2022}, \citet{HebertWoodford2023}, and \citet{BloedelZhong2025} allow decision makers to choose what to observe at each step. We study a specific technology in which the producer samples a given experiment at a constant cost per observation. This restriction allows us to derive an explicit formula for the large-scale unit cost of producing any target experiment and to characterize exactly which cost functions this technology generates.

More broadly, our paper relates to the literature on optimal dynamic information acquisition, which characterizes optimal sampling, attention, and stopping policies in particular environments \citep{MoscariniSmith2001,FudenbergStrackStrzalecki2018,CheMierendorff2019,LiangMuSyrgkanis2022,Zhong2022}. Rather than solving for an optimal acquisition policy in a given environment, we ask when one experiment is at least as valuable as another in every decision problem under optimal sequential sampling, with sufficiently large budgets or small costs.

The rest of the paper is organized as follows. \Cref{sec:orders} presents a motivating example. \Cref{sec:model} introduces the model and two orders. \Cref{sec:results} states the ranking theorem, compares it with fixed-sample orders, and proves necessity. \Cref{sec:reproduction} develops the exact reproduction theorem and proves sufficiency. \Cref{sec:costs} characterizes large-scale information-production costs, and \Cref{sec:conclusion} concludes.

\section{A Motivating Example}\label{sec:orders}
Consider a binary hypothesis testing problem. The state is $0$ or $1$ with a uniform prior, and a decision maker gets a payoff of $1$ if she matches the state and $0$ otherwise. Before taking the action, she can buy labels from either the crowd pool or the expert pool. Each label from the crowd pool reports the correct state with probability $0.75$. Each label from the expert pool comes from an expert with probability $0.1$; otherwise, it is blank. An expert label reports the correct state with probability $0.999$. Labels from either pool are independent conditional on the state. 

Consider two settings. With a ``fixed sample size,'' the decision maker buys twenty labels and guesses the state after observing all of them.\footnote{This parallels the fixed-sample framework of \citet{MoscariniSmith2002}, where the sample size is chosen before observing the signals.} With ``sequential sampling,'' she decides after each label whether to buy another or stop and guess the state. She can choose any strategy whose ex ante expected number of labels does not exceed twenty. With a fixed sample size, the optimal strategy is to guess the state reported by the majority of the twenty labels; with sequential sampling, the optimal strategy is to stop when the number of labels reporting one state exceeds the number reporting the other by a specified amount,\footnote{Strictly speaking, meeting the expected-sample budget exactly may require randomization between adjacent stopping thresholds, which we abstract from for simplicity.}
then guess the state with more labels. The amount is chosen separately for each pool to yield twenty labels in expectation ex-ante.\footnote{With the crowd pool, she stops when one state leads by about ten labels; with the expert pool, she stops when one state leads by about two expert labels. The proposed stopping rules are optimal in the sense that they minimize the expected number of labels in each state among all rules with weakly lower error probabilities in both states \citep{WaldWolfowitz1948}.} \Cref{tab:example} reports the error probabilities. The crowd pool has a lower error probability with a fixed sample size, but the expert pool has a lower error probability with sequential sampling.

\begin{table}[ht]
\centering
\caption{Error probability}\label{tab:example}
\begin{tabular}{@{}lcc@{}}
\toprule
 & Crowd pool & Expert pool\\
\midrule
Fixed sample size (twenty labels) & $0.9\%$ & $6.1\%$\\
Sequential sampling (twenty labels on average) & $0.0017\%$ & $0.0005\%$\\
\bottomrule
\end{tabular}
\end{table}

To see why the ranking reverses between the two settings, define the evidence in favor of the true state as the log-likelihood ratio of the labels received under the true state against the other state. A correct label raises the evidence, an incorrect label lowers it, and a blank label leaves it unchanged. The decision maker's payoff depends on the distribution of her evidence at stopping. With a fixed sample size, she must stop after twenty labels, even if the evidence remains weak or misleading. The expert pool supplies evidence in rare, large steps. If none of the twenty labels comes from an expert, the decision maker has no evidence and must guess. The crowd pool supplies evidence in smaller steps with every label, so after twenty labels its evidence is less likely to be weak or misleading. The outcomes that lead to an error are shown in red in \Cref{fig:paths}.

\begin{figure}[ht]
\centering
\includegraphics[width=\textwidth]{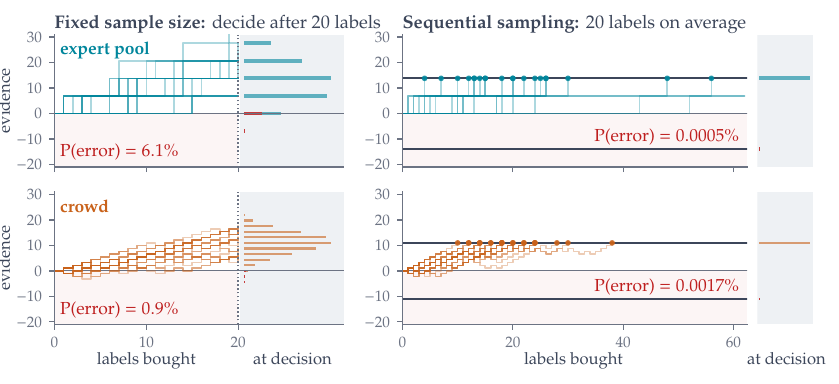}
\caption{Simulated paths of the evidence in favor of the true state. Left: the decision maker buys twenty labels. Right: the decision maker stops when the evidence reaches either horizontal line; dots mark the stopping points. Gray panels: evidence at the time of the decision, with outcomes that lead to an error shown in red.}\label{fig:paths}
\end{figure}

With sequential sampling, instead of stopping at a fixed number of observations, the decision maker stops when the evidence reaches a fixed level in favor of either state, called the evidence threshold, as shown in \Cref{fig:paths}. She buys more labels while the evidence remains weak. Relative to a fixed sample of twenty labels, the stopping rules shift labels from histories where the evidence becomes strong early to histories where gathering strong evidence takes longer. The comparison therefore turns on how high an evidence threshold each pool reaches with twenty labels on average.

The expert pool reaches a higher evidence threshold because the expert pool supplies more expected evidence per label. By Wald's identity, the expected evidence at stopping equals the expected evidence per label times the expected number of labels. Expert labels are rare, but each expert label carries enough evidence to more than offset the frequency of blank labels. With twenty labels on average, the expert pool therefore reaches a higher evidence threshold and has a lower error probability.

This example shows that replacing a fixed sample size under an ex-post budget with sequential sampling under an ex-ante budget on the expected number of observations can reverse the ranking of experiments. The reversal is not specific to this hypothesis test with twenty labels: as shown below, it extends to every nondegenerate decision problem when the decision maker can sample extensively. The next sections formalize comparisons across decision problems and characterize when one experiment is more valuable than another under sequential sampling under an ex-ante budget.

\section{Model}\label{sec:model}
In this section, we first formalize an environment in which a decision maker sequentially samples an experiment before choosing an action. We then introduce two ways of ranking experiments across decision problems, depending on whether the decision maker faces a budget on the expected number of observations, as in \Cref{sec:orders}, or pays a cost per observation without a budget constraint.

\subsection{Environment}
The state space $\Theta$ is finite, with $K\ge2$ elements. An experiment $E=(E_i)_{i\in\Theta}$ specifies a probability distribution $E_i$ in each state $i$.\footnote{Throughout the paper, all signals take values in a fixed uncountable standard Borel space, such as $[0,1]$.}
Throughout, the \emph{source experiments}, i.e., experiments available for sampling, satisfy \Cref{ass:regular}.
\begin{assumption}\label{ass:regular}
For all distinct states $i$ and $j$, the distributions $E_i$ and $E_j$ are distinct and mutually absolutely continuous. The log-likelihood ratio $\log(dE_i/dE_j)$ has a finite second moment under $E_i$.
\end{assumption}


A decision maker repeatedly samples an experiment $E$. Conditional on state $i$, the signals $X_1,X_2,\ldots$ are independent and identically distributed according to $E_i$. Before each observation, including the first, the decision maker decides whether to stop based on the signal realizations already observed and, possibly, randomization independent of the state and the signals. The number of observations is thus a stopping time $T$ with respect to the information generated by the signals and the independent randomization. The notation $\Prob_i^E$ and $\E_i^E$ denotes probabilities and expectations in state $i$ when sampling $E$. 

A \emph{stopping policy}, or simply a policy, is a pair $(T,Z)$ consisting of a stopping time $T$ with $\E_i^E T<\infty$ in every state $i$ and a signal $Z$. The realization of $Z$ is determined at time $T$ by the observed sample realizations and possibly independent randomizations.
The \emph{terminal experiment} generated by a policy $(T,Z)$ is
\[
H(T,Z):=\bigl(\Law_i^E(Z)\bigr)_{i\in\Theta},
\]
where $\Law_i^E(Z)$ denotes the distribution of $Z$ in state $i$.\footnote{The state-contingent distributions of a terminal experiment are mutually absolutely continuous, but need not be distinct or have log-likelihood ratios with finite second moments. Thus, a terminal experiment need not satisfy \Cref{ass:regular}.}


\subsection{Large-budget dominance}\label{subsec:lbd}

A decision problem $\mathcal D=(\mu,\mathcal A,u)$ consists of a prior $\mu\in\Delta(\Theta)$, a finite action set $\mathcal A$, and a payoff function $u:\mathcal A\times\Theta\to\R$. The decision maker samples experiment $E$ sequentially and decides when to stop and which action to take based on the signal realizations observed. The terminal action is a signal taking values in $\mathcal A$, so the decision maker's strategy can be described as a stopping policy $(T,a)$. 

Given a budget $n\ge0$ on the expected number of observations, the decision maker's value\footnote{Requiring $\E_i^E T<\infty$ in states with zero prior probability affects neither $V_n$ nor the value $W_c$ defined in ~\Cref{subsec:scd}. To see why, truncate a policy at a large fixed date and take an arbitrary action if the policy has not stopped. In each state with positive prior probability, the change in expected payoff is bounded by a constant times the probability that the policy has not stopped by that date. This probability tends to zero as the truncation date grows, and truncation does not increase the number of observations.} is
\[
 V_n(E;\mathcal D):=\sup\Bigl\{\sum_{i\in\Theta}\mu_i\E_i^E[u(a,i)]:(T,a)\text{ is a policy with }\sum_{i\in\Theta}\mu_i\E_i^E T\le n\Bigr\}.
\]

\begin{definition}[Large-budget dominance]\label{def:lbd}
Experiment $F$ \emph{large-budget dominates} experiment $G$ if there is $\bar n<\infty$ such that $V_n(F;\mathcal D)\ge V_n(G;\mathcal D)$ for every decision problem $\mathcal D$ and every $n\ge\bar n$.
\end{definition}


Large-budget dominance is the sequential-sampling analogue of the large-sample Blackwell order studied by \citet{MuPomattoStrackTamuz2021}. It requires that, beyond some common expected-sample budget threshold, one experiment yields at least as much value as another in every decision problem. We highlight two features. First, a single budget threshold $\bar n$ applies uniformly to all decision problems. If $F$ large-budget dominates $G$, a decision maker whose budget exceeds $\bar n$ can rank $F$ weakly above $G$ without knowing her prior, available actions, or payoffs. This robustness matters when an experiment must be chosen before the decision problem is known or when a single experiment serves many decision problems.

Second, the budget is imposed ex ante on the expected number of observations.\footnote{Expected sample size is a standard measure of sampling effort in sequential analysis \citep{WaldWolfowitz1948,Lorden1977}.} This constraint on the expected budget has a natural interpretation when many independent, identical decision problems share a hard aggregate budget. Suppose $M$ such problems share a budget of $Mn$ observations. As $M$ grows, the optimal average payoff under the hard budget converges to the value of a single problem under an expected-sample budget of $n$.\footnote{To see this, apply to each problem a nearly optimal policy with an expected sample size slightly below the average budget $n$. By the law of large numbers, the probability of exhausting the aggregate budget vanishes. Stopping unfinished problems if the budget is exhausted ensures feasibility with negligible payoff loss. Conversely, consider any joint policy that allocates the shared budget across the $M$ problems, and let $n_k$ denote the expected number of observations allocated to problem $k$. Independence implies that the expected payoff from problem $k$ is at most $V_{n_k}(E;\mathcal D)$. Since the decision maker can randomize between policies at time zero, $V_n(E;\mathcal D)$ is concave in $n$. Thus, for $M$ problems with a total budget of $Mn$, the average payoff is at most $\frac{1}{M}\sum_{k=1}^M V_{n_k}(E;\mathcal D)\le V_n(E;\mathcal D)$.} 
Thus, if $F$ large-budget dominates $G$, then for every decision problem $\mathcal D$ and every $n\ge\bar n$, the limiting optimal average payoff from $M$ independent copies of $\mathcal D$ sharing a hard budget of $Mn$ is at least as high under $F$ as under $G$, as $M\to\infty$.


\subsection{Small-cost dominance}\label{subsec:scd}
An alternative formulation is to consider a decision maker who pays a cost per observation and faces no budget constraint. At a cost $c>0$ per observation, the decision maker's value is
\[
W_c(E;\mathcal D):=\sup_{(T,a)}\sum_{i\in\Theta}\mu_i\E_i^E[u(a,i)-cT],
\]
where the supremum is over stopping policies $(T,a)$.
\begin{definition}[Small-cost dominance]\label{def:scd}
Experiment $F$ \emph{small-cost dominates} experiment $G$ if, for every decision problem $\mathcal D$, there exists $\bar c_{\mathcal D}>0$ such that
\[
W_c(F;\mathcal D)\ge W_c(G;\mathcal D)
\qquad\text{for all }c\in(0,\bar c_{\mathcal D}).
\]
\end{definition}

The problem of sequential sampling with a cost per observation is classical in sequential analysis \citep{ArrowBlackwellGirshick1949,Chernoff1959,KieferSacks1963}. Here, we compare experiments by the values they deliver when sampling costs are small. An important feature of \Cref{def:scd} is that the cost threshold $\bar c_{\mathcal D}$ may depend on the decision problem. This dependence is natural because whether a sampling cost is small depends on the stakes involved. To see why, let $\beta\mathcal D$ denote the decision problem obtained by multiplying all payoffs in $\mathcal D$ by $\beta>0$. Then
\[
W_c(E;\beta\mathcal D)=\beta W_{c/\beta}(E;\mathcal D).
\]
Thus, even a small cost $c$ can correspond to a large effective cost $c/\beta$ when the stakes are low. With unrestricted payoff scales, requiring a common positive cost threshold across all decision problems would imply dominance at every positive cost, rather than merely at sufficiently small costs.

\begin{remark}\label{rem:scd-alternatives}
The dependence on payoff scales can be removed by restricting attention to classes of decision problems in which the value of information is uniformly bounded away from zero and the maximum loss from a suboptimal action is uniformly bounded. 
We can then strengthen small-cost dominance by requiring a common cost threshold within each such class. Another strengthening allows discounting of terminal payoffs and sampling costs, and requires dominance when both the cost per observation and the discount rate are sufficiently small. Despite being stronger, both comparisons yield the same ranking as small-cost dominance for generic pairs of experiments. See Online Appendix~\ref{app:alternative-orders}.
\end{remark}

The budget-constrained and flow-cost problems are related by the identity
\begin{equation}\label{eq:duality}
W_c(E;\mathcal D)=\sup_{n\ge0}\bigl\{V_n(E;\mathcal D)-cn\bigr\}.
\end{equation}
A policy with an expected number of observations equal to $n$ yields an expected payoff net of costs of at most $V_n(E;\mathcal D)-cn$. Conversely, every policy feasible for $V_n(E;\mathcal D)$ incurs an expected sampling cost of at most $cn$. Taking the supremum over feasible policies establishes \eqref{eq:duality}.
Therefore, the large-budget and small-cost orders provide two different but related ways to compare experiments. As we show in the next section, the two resulting rankings generically coincide and admit a common characterization.

\section{Ranking Experiments}\label{sec:results}

We now characterize large-budget and small-cost dominance. We then contrast these orders with existing fixed-sample orders, showing how endogenous stopping can reverse the ranking of experiments.

\subsection{The ranking theorem}\label{subsec:theorem1}

For distinct states $i$ and $j$, the \emph{Kullback--Leibler (KL) divergence} of $E_i$ from $E_j$ is
\[
 I_{ij}(E):=\int\log\frac{dE_i}{dE_j}\,dE_i.
\]
In state $i$, $I_{ij}(E)$ measures the expected evidence per observation in favor of $i$ against $j$.
Under \Cref{ass:regular}, every KL divergence $I_{ij}(E)$ is positive and finite. 
We say experiments $F$ and $G$ are \emph{non-tied} if $I_{ij}(F)\ne I_{ij}(G)$ for every ordered pair of distinct states $i$ and $j$.

\begin{theorem}[Ranking of experiments]\label{thm:generic-equivalence}
Let $F$ and $G$ be non-tied experiments. The following statements are equivalent:
\begin{enumerate}[label=\textup{(\roman*)}]
 \item $F$ large-budget dominates $G$;
 \item $F$ small-cost dominates $G$;
 \item $I_{ij}(F)>I_{ij}(G)$ for every ordered pair of distinct states $i$ and $j$.
\end{enumerate}
\end{theorem}

\Cref{thm:generic-equivalence} establishes that the two orders generically coincide: one experiment large-budget dominates another if and only if it small-cost dominates the other. Moreover, both orders admit a simple characterization in terms of \emph{KL dominance}, the requirement in condition (iii) that one experiment have a larger KL divergence in every ordered state direction. Thus, ranking experiments by their optimal values across all decision problems reduces to checking $K(K-1)$ pairwise inequalities.\footnote{The non-tiedness assumption is needed only for necessity. For any pair of experiments, strict KL dominance implies both orders, and either order implies the weak inequalities $I_{ij}(F)\ge I_{ij}(G)$. When some KL divergences are equal, these weak inequalities need not suffice for large-budget dominance; see \Cref{rem:tight}.}

The role of KL divergence in \Cref{thm:generic-equivalence} can be understood from how endogenous stopping uses evidence. Recall the expert and crowd pools in \Cref{sec:orders}: although expert labels are often uninformative, an informative expert label provides substantial evidence. Sequential sampling allows the decision maker to continue when the evidence remains inconclusive and stop once it becomes sufficiently decisive. The relevant consideration is therefore how much evidence can be accumulated for a given expected-sample budget. When state $i$ is true, each observation from $E$ provides, on average, $I_{ij}(E)$ units of log-likelihood evidence against state $j$. By Wald's identity, the expected evidence accumulated by stopping equals $I_{ij}(E)$ multiplied by the expected sample size. This suggests why pairwise KL dominance is necessary: an experiment that accumulates evidence more slowly in some direction may be less valuable in a decision problem where that distinction is important. The more difficult part is to show that higher evidence rates in every direction are also sufficient for large-budget dominance, which requires a common budget threshold across all decision problems.

We prove \Cref{thm:generic-equivalence} through the cycle of implications (i)$\Rightarrow$(ii)$\Rightarrow$(iii)$\Rightarrow$(i). The first implication follows directly from \eqref{eq:duality}: when each observation costs $c$, the decision maker can equivalently first choose an expected-sample budget $n$ and then solve the corresponding budget-constrained problem, net of the cost $cn$. Fix a decision problem $\mathcal D$ in which information has positive value, and let $n_c$ be an optimal budget for $G$. As $c\downarrow0$, we have $n_c\to\infty$. Thus, if $F$ large-budget dominates $G$ with threshold $\bar n$, then $n_c\ge\bar n$ for all sufficiently small $c$, and
\[
W_c(F;\mathcal D)
\ge V_{n_c}(F;\mathcal D)-cn_c
\ge V_{n_c}(G;\mathcal D)-cn_c
= W_c(G;\mathcal D).
\]

The second implication follows from classical sequential-testing asymptotics. Write $V^*(\mathcal D):=\sum_{k\in\Theta}\mu_k\max_{a\in\mathcal A}u(a,k)$ for the full-information value of a decision problem $\mathcal D$. Consider a binary decision problem with prior supported on states $i$ and $j$, with $\mu_i,\mu_j>0$, and two actions, each uniquely optimal in a different state. A classical result of \citet{Chernoff1959} gives, as $c\downarrow0$,
\begin{equation}\label{eq:binary-small-cost-loss}
 V^*(\mathcal D)-W_c(E;\mathcal D)
 =c\log\frac{1}{c}\left(\frac{\mu_i}{I_{ij}(E)}+\frac{\mu_j}{I_{ji}(E)}\right)(1+o(1)).
\end{equation}
Intuitively, optimal sequential testing balances the marginal benefit of reducing decision errors against the cost $c$ of an additional observation. As $c\downarrow0$, this leads the decision maker to continue until the posterior probability of error is of order $c$, requiring log-likelihood evidence of order $\log(1/c)$. Under state $i$, accumulating this evidence requires approximately $\log(1/c)/I_{ij}(E)$ observations, so the leading loss is the expected sampling cost, weighted by the prior. If $F$ small-cost dominates $G$, the leading coefficient for $F$ cannot exceed that for $G$ in any such binary problem. Letting the prior probability on state $i$ approach one then yields $I_{ij}(F)\ge I_{ij}(G)$. Since every ordered pair can be isolated in this way, small-cost dominance implies all the weak pairwise KL inequalities, which become strict under non-tiedness.

The remaining implication, (iii)$\Rightarrow$(i), is more challenging.\footnote{On the other hand, the weaker implication (iii)$\Rightarrow$(ii) can be established using classical small-cost asymptotics. \citet{KieferSacks1963} extend expansions of the form \eqref{eq:binary-small-cost-loss} to general decision problems. Strict pairwise KL dominance then reduces the leading loss coefficient in every fixed decision problem with positive value of information, yielding small-cost dominance.} Classical sequential-testing asymptotics provide first-order comparisons for each fixed decision problem, whereas large-budget dominance requires a single budget threshold that works across all decision problems. Our second main result, developed in \Cref{sec:reproduction}, provides the uniformity needed to bridge this gap. Before turning to that result, we compare our orders with existing fixed-sample orders to clarify how endogenous stopping changes the criteria for ranking experiments.

\subsection{Fixed versus sequential sampling}\label{subsec:fixed}
It is instructive to contrast our rankings with two fixed-sample orders. The MS order \citep{MoscariniSmith2002} compares experiments at sufficiently large sample sizes, allowing the threshold to depend on the decision problem. The MPST order \citep{MuPomattoStrackTamuz2021} instead requires that $F^{\otimes n}$ Blackwell dominate $G^{\otimes n}$ for all sufficiently large $n$, so the threshold is common to all decision problems.  \Cref{tab:rankings} summarizes the statistical criteria governing these comparisons and our two orders under sequential sampling.

\begin{table}[ht]
\centering
\caption{Criteria for comparing experiments in large samples}\label{tab:rankings}
\begin{tabular}{@{}lll@{}}
\toprule
 & Threshold may depend & One threshold for all \\
 & on the decision problem & decision problems \\
\midrule
Fixed sample size & Chernoff information & R\'enyi divergences of all orders \\
 & \citep{MoscariniSmith2002} & \citep{MuPomattoStrackTamuz2021} \\[4pt]
Sequential sampling & Pairwise KL divergences & Pairwise KL divergences \\
 & (\Cref{thm:generic-equivalence}) & (\Cref{thm:generic-equivalence}) \\
\bottomrule
\end{tabular}
\par\medskip
\begin{minipage}{0.88\textwidth}
\footnotesize  The MPST characterization applies to generic pairs of bounded binary experiments. In the sequential-sampling row, the left and right entries correspond to small-cost and large-budget dominance, respectively.
\end{minipage}
\end{table}

\begin{remark}[Alternative formulations]
The sequential entries use costs for the problem-by-problem comparison and budgets for the uniform comparison. This choice is inessential: one can instead impose a common small-cost threshold on classes of decision problems with normalized stakes, as noted in \Cref{rem:scd-alternatives}, or allow a problem-dependent large-budget threshold. For non-tied pairs, both alternatives are again characterized by pairwise KL dominance (Online Appendix~\ref{app:alternative-orders}).
\end{remark}

For binary experiments, the three statistical criteria can be understood using a single object. Define the log Hellinger transform
\[
 \psi_E(s):=\log\int\left(\frac{dE_1}{dE_0}\right)^s dE_0,
 \qquad s\in\mathbb R,
\]
where the integral may be infinite. The function is convex and satisfies $\psi_E(0)=\psi_E(1)=0$. 
For nondegenerate binary decision problems, the MS comparison is governed by Chernoff information $\mathrm{Ch}(E):=-\min_{s\in[0,1]}\psi_E(s)$, the depth of the curve. The MPST order instead requires larger R\'enyi divergences of every positive order in both state directions.\footnote{For $s>0$ and $s\ne1$, the order-$s$ R\'enyi divergences for $E_1$ against $E_0$ and $E_0$ against $E_1$ are $\psi_E(s)/(s-1)$ and $\psi_E(1-s)/(s-1)$, respectively. At $s=1$, they coincide with the corresponding KL divergences. Thus, the MPST criterion implies both the MS order and our two sequential orders.} In terms of the Hellinger transform, this requires the dominating experiment's curve to lie strictly below the other's on $(0,1)$ and strictly above it outside $[0,1]$. By contrast, our two sequential orders depend only on the endpoint slopes, $-\psi_E'(0)=I_{01}(E)$ and $\psi_E'(1)=I_{10}(E)$. \Cref{fig:psi} illustrates these distinctions for the expert and crowd pools.\footnote{For an interactive illustration of how the curves vary with expert and crowd accuracy and the share of informative expert labels, see \url{https://theory-demos.github.io/KL-demo/}.}

\begin{figure}[ht]
\centering
\includegraphics[width=0.6\textwidth]{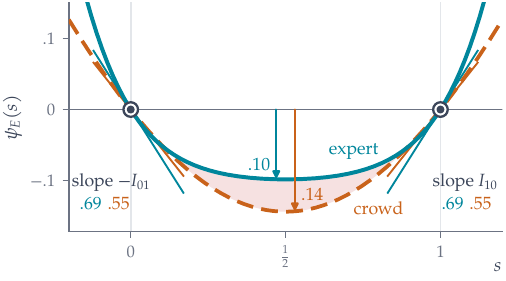}
\vspace{-3mm}
\caption{The log Hellinger transforms for the expert pool (solid) and crowd pool (dashed) in \Cref{sec:orders}. The crowd pool has larger Chernoff information (the depth of the curve), while the expert pool has larger KL divergences (the endpoint slopes). The curves cross, so neither pool dominates the other in the MPST order.}\label{fig:psi}
\end{figure}

The contrast reflects how the two sampling technologies handle inconclusive evidence. With a fixed sample size, the decision maker must act at the prescribed horizon, even if the evidence remains inconclusive. In a fixed nondegenerate binary decision problem, the error probabilities decay at the Chernoff rate; a uniform fixed-sample comparison must also cover decision problems whose priors or payoffs vary with the sample size, so the full Rényi spectrum becomes relevant. Under sequential sampling, by contrast, the decision maker can continue sampling on inconclusive histories, so the relevant first-order criterion is the expected rate of evidence accumulation, measured by directed KL divergences. We also note that, unlike under a fixed sample size, the distinction between problem-dependent and uniform thresholds does not change the ranking under sequential sampling, generically. This follows from a stronger result established in \Cref{sec:reproduction}: strict KL dominance yields a first-order saving in expected sampling cost that outweighs a lower-order overhead uniformly over stopping policies, allowing a common budget threshold without additional divergence conditions.

These different criteria can reverse the ranking of experiments. A transparent way to see this is through \emph{dilution}. For an experiment $E$ and $\alpha\in(0,1)$, let $E^\alpha$ deliver an observation of $E$ with probability $\alpha$ and a distinguishable uninformative signal otherwise. The following proposition shows that dilution can generate systematic ranking reversals.

\begin{proposition}[Ranking reversal under dilution]\label{prop:reversal}
Let $G$ be any experiment satisfying \Cref{ass:regular}. There exists $\bar\alpha_G\in(0,1)$ such that, for every $\alpha\in(0,\bar\alpha_G)$, $G$ strictly dominates $E^\alpha$ in the MS order, regardless of $E$.\footnote{That is, in every decision problem with positive value of information, $G$ yields a strictly higher optimal value than $E^\alpha$ for all sufficiently large fixed sample sizes.} Nevertheless, for every $\alpha\in(0,1)$, there exists an experiment $E$ such that $E^\alpha$ large-budget dominates and small-cost dominates $G$.
\footnote{When $\alpha<\bar\alpha_G$, any such $E^\alpha$ and $G$ are Blackwell-incomparable at every fixed sample size, as the Chernoff and pairwise KL comparisons favor opposite directions.}
\end{proposition}

To understand the proposition, fix an experiment $G$. For any experiment $E$, its dilution $E^\alpha$ delivers an uninformative signal with probability $1-\alpha$. Consequently, its pairwise Chernoff information is bounded above by $-\log(1-\alpha)$, regardless of how informative $E$ is. We can therefore choose $\alpha$ sufficiently small so that $G$ has strictly larger pairwise Chernoff information than every $E^\alpha$, and hence $G$ eventually yields strictly higher value than each $E^\alpha$ in every decision problem with positive value of information.
In contrast, take $E=G^{\otimes m}$, so that each observation from $E$ consists of $m$ conditionally independent observations from $G$. Its dilution $E^\alpha$ delivers this batch with probability $\alpha$ and an uninformative signal otherwise. Since $I_{ij}(E^\alpha)=\alpha m I_{ij}(G)$, choosing $m>1/\alpha$ gives $E^\alpha$ strictly larger KL divergences in every direction, and hence dominance under both sequential orders. For sufficiently small $\alpha$, these opposing Chernoff and KL comparisons also rule out fixed-sample Blackwell dominance in either direction.

The motivating example in \Cref{sec:orders} illustrates this reversal. For the crowd pool, the threshold in the proposition is approximately $0.13$, whereas informative expert labels arrive with probability $\alpha=0.1$. Thus, regardless of how accurate the expert labels are, the crowd pool yields strictly higher value in every decision problem with positive value of information at sufficiently large fixed sample sizes. By contrast, no matter how rarely informative expert labels arrive, sufficiently high expert accuracy makes the expert pool dominate the crowd pool under both sequential orders.

\section{Exact Reproduction}\label{sec:reproduction}

In this section, we develop a production characterization of large-budget dominance and use it to establish the sufficiency of KL dominance. We first show that large-budget dominance is equivalent to an inclusion between the sets of terminal experiments attainable from the two sources under a common expected-sample budget. We then introduce our main technical result, \Cref{thm:conversion}, and use it to establish this inclusion for all sufficiently large budgets.

\subsection{Production dominance}\label{subsec:terminal}
Rather than evaluating a source experiment through its value in a particular decision problem, we can ask which terminal experiments it can produce through sequential sampling. To formalize this perspective, for a prior $\mu$ and a budget $n$, let
\[
 \mathcal P_n^\mu(E):=\Bigl\{H(T,Z):(T,Z)\text{ is a policy for $E$ with }\sum_{i\in\Theta}\mu_i\E_i^ET\le n\Bigr\}
\]
be the set of terminal experiments that $E$ generates within the budget. This set contains every garbling of its elements, because a policy can garble its signal with independent randomization; and it is convex, because a policy can randomize between two policies at time zero.

\begin{definition}[Production dominance]\label{def:production}
Experiment $F$ \emph{production dominates} experiment $G$ if there is $\bar n<\infty$ such that $\mathcal P_n^\mu(G)\subseteq\mathcal P_n^\mu(F)$ for every prior $\mu$ and every $n\ge\bar n$.
\end{definition}

Production dominance compares source experiments as technologies for producing information, independently of how that information will be used. It requires that, beyond a common budget threshold for all priors and target experiments, $F$ can exactly reproduce the state-contingent distribution of any terminal experiment attainable from $G$ within the same expected-sample budget. Thus, at sufficiently large budgets, any decision maker can achieve at least as much value with $F$ as with $G$, regardless of her prior or decision problem. The following proposition shows that the converse also holds: production dominance is equivalent to large-budget dominance.

\begin{proposition}[Value and production]\label{prop:inclusion}
Experiment $F$ large-budget dominates experiment $G$ if and only if $F$ production dominates $G$.
\end{proposition}

\Cref{prop:inclusion} is the analogue of Blackwell's theorem for sequential sampling. Just as Blackwell's theorem characterizes dominance across decision problems through garbling, \Cref{prop:inclusion} characterizes large-budget dominance through the terminal experiments that each source experiment can generate within the same budget. As discussed above, production dominance immediately implies large-budget dominance. For the converse, when the prior $\mu$ has full support, the standard separation argument underlying Blackwell's theorem shows that every $H\in\mathcal P_n^\mu(G)$ can also be generated from $F$ within the same budget $n$. For priors without full support, we perturb the prior to have full support and obtain policies for $F$ that reproduce $H$ under budgets converging to $n$. A compactness argument then yields a limiting policy that reproduces $H$ within the original budget.

By \Cref{prop:inclusion}, it suffices to show that KL dominance implies production dominance. This presents two challenges. First, although each terminal experiment in $\mathcal P_n^\mu(G)$ can be reproduced from $F$ with sufficiently many observations, it is unclear whether the required number can be uniformly bounded across all such experiments. Second, even with a uniform bound, we must show that it eventually falls below $n$, despite the fact that $\mathcal P_n^\mu(G)$ expands with $n$. Our \Cref{thm:conversion}, introduced below, resolves both challenges by providing a sharp first-order bound on the expected number of observations needed to reproduce any target experiment, with a remainder uniform over the targets.

\subsection{The reproduction theorem}\label{subsec:theorem2}

For experiments $F$ and $G$ and a state $i$, define the conversion rate
\[
 \rho_i(F\to G):=\max_{j\ne i}\frac{I_{ij}(G)}{I_{ij}(F)}.
\]
The ratio measures the number of observations of $F$ needed to match the average evidence provided by one observation of $G$ against $j$. We call an alternative $j$ that attains the maximum a \emph{bottleneck} of $G$ for $F$ in state $i$.

\begin{theorem}[Exact reproduction]\label{thm:conversion}
Let $F$ and $G$ satisfy \Cref{ass:regular},
except that $G_i=G_j$ is allowed. There is a constant $A_{F,G}<\infty$ such that, for every policy $(T,Z)$ for $G$, there is a policy $(N,Y)$ for $F$ that generates the same terminal experiment, $H(N,Y)=H(T,Z)$, and satisfies
\begin{equation}\label{eq:uniform-conversion}
 \E_i^FN\le\rho_i(F\to G)\,\E_i^GT+A_{F,G}\sqrt{\E_i^GT}\qquad\text{for every state }i.
\end{equation}
\end{theorem}

\Cref{thm:conversion} shows that any terminal experiment generated by sequential sampling from $G$ can be reproduced exactly from $F$, with an upper bound on the expected number of observations required in each state. Two features of this bound are important. First, the conversion rate $\rho_i(F\to G)$ is sharp: no smaller first-order coefficient works uniformly over all policies for $G$. To see this, consider a policy that takes $n$ observations of $G$ and reports the entire sample, generating $G^{\otimes n}$. Any policy $(N,Y)$ for $F$ that reproduces this experiment must satisfy, by data processing and Wald's identity, $nI_{ij}(G)\le I_{ij}(F)\E_i^FN$ for every $j\ne i$.\footnote{Related lower bounds on expected sample sizes are established by \citet{Simons1967} for sequential multihypothesis testing under error constraints, and by \citet{GreenshteinTorgersen1997} for comparisons of stopped experiments generated from a common sampling source.}
Hence,
\begin{equation}\label{eq:product-lower}
\E_i^FN\ge n\,\rho_i(F\to G).
\end{equation}
Second, the coefficient $A_{F,G}$ of the square-root remainder depends only on the two source experiments, not on the policy $(T,Z)$. This uniformity is essential for establishing production dominance with a common budget threshold.

We now show how \Cref{thm:conversion} resolves the two challenges in establishing production dominance. Suppose that strict KL dominance holds, so that $\rho:=\max_i\rho_i(F\to G)<1$. Fix any prior $\mu$, budget $n$, and terminal experiment $H\in\mathcal P_n^\mu(G)$. Let $(T,Z)$ be a policy for $G$ that generates $H$, with $t_i:=\E_i^GT$ satisfying $\sum_i\mu_it_i\le n$. \Cref{thm:conversion} yields a policy $(N,Y)$ for $F$ that reproduces $H$ exactly and satisfies
\[
\sum_{i\in\Theta}\mu_i\E_i^FN
\le \rho\sum_{i\in\Theta}\mu_it_i
+A_{F,G}\sum_{i\in\Theta}\mu_i\sqrt{t_i}
\le \rho n+A_{F,G}\sqrt n,
\]
where the second inequality follows from Jensen's inequality and the budget constraint. Since $A_{F,G}$ depends only on the two source experiments, this bound applies uniformly to all terminal experiments in $\mathcal P_n^\mu(G)$. Moreover, the bound eventually falls below the original budget $n$, even as $\mathcal P_n^\mu(G)$ expands: because $\rho<1$, we have $\rho n+A_{F,G}\sqrt n\le n$ whenever $n$ is sufficiently large. This threshold depends only on $F$ and $G$, so $\mathcal P_n^\mu(G)\subseteq\mathcal P_n^\mu(F)$ for every prior $\mu$ and every $n\ge\bar n$. By \Cref{prop:inclusion}, $F$ large-budget dominates $G$. This completes the sufficiency proof of \Cref{thm:generic-equivalence}. \Cref{fig:flow} summarizes the proof of the ranking theorem.

\begin{figure}[ht]
\vspace{3mm}
\centering
\begin{tikzpicture}[
  pf/.style={draw=darkblue!55,rounded corners=3pt,fill=darkblue!4,align=center,font=\small,
             minimum height=1.0cm,minimum width=2.35cm,inner sep=3pt},
  pfarr/.style={-{Stealth[length=2.6mm,width=2.1mm]},line width=.9pt,darkblue!70},
  pflab/.style={font=\footnotesize,text=muted,align=center,inner sep=1pt}]
 \node[pf] (pr) at (0,0) {production\\ dominance};
 \node[pf] (lb) at (4.6,0) {large-budget\\ dominance};
 \node[pf] (sc) at (9.2,0) {small-cost\\ dominance};
 \node[pf,fill=Fcol!10,draw=Fcol!70] (kl) at (13.8,0) {KL\\ dominance};
 \draw[pfarr,{Stealth[length=2.6mm,width=2.1mm]}-{Stealth[length=2.6mm,width=2.1mm]}] (pr) -- node[above=3pt,pflab] {Proposition 2} (lb);
 \draw[pfarr] (lb) -- node[above=3pt,pflab] {budget--cost\\ duality} (sc);
 \draw[pfarr] (sc) -- node[above=3pt,pflab] {Chernoff\\ expansion} (kl);
 \draw[pfarr,draw=Fcol,rounded corners=10pt] (kl.south) -- ++(0,-.55) -| node[pos=.25,below=2pt,pflab,text=Fcol] {Theorem 2} (pr.south);
\end{tikzpicture}
\caption{The proof of \Cref{thm:generic-equivalence}. Each implication holds for every pair of experiments, except that small-cost dominance implies the strict inequalities of KL dominance only for non-tied pairs; for other pairs, it implies the weak inequalities.}\label{fig:flow}
\end{figure}

\begin{remark}\label{rem:tight}
The square-root remainder in \Cref{thm:conversion} cannot, in general, be improved to a smaller order. Online Appendix~\ref{app:multiple-bottlenecks} constructs three-state experiments $F$ and $G$ with $I_{ij}(F)\ge I_{ij}(G)$ for every ordered pair, with equality whenever either state is $0$. For arbitrarily large budgets $n$, there exists a policy $(T,Z)$ for $G$ with $\E_0^GT\le n$, such that every policy $(N,Y)$ for $F$ reproducing its terminal experiment satisfies $\E_0^FN\ge n+\kappa\sqrt n$ for some constant $\kappa>0$. In particular, under the prior concentrated on state $0$, the terminal experiment belongs to $\mathcal P_n^\mu(G)$ but not to $\mathcal P_n^\mu(F)$ for arbitrarily large $n$. Hence, weak KL dominance does not imply production dominance or large-budget dominance. By contrast, when each state's conversion rate has a unique maximizer, the remainder can be reduced to a constant uniform over all target stopping policies, as shown in Online Appendix~\ref{app:unique-bottleneck}.
\end{remark}

\subsection{How exact reproduction works}\label{subsec:construction}

We sketch the constructive proof of \Cref{thm:conversion}, which proceeds in four steps. First, we propose a signal realization $z$ of a target experiment $H$, together with an auxiliary state $k$, and use observations from $F$ to determine whether to release $z$. Second, we calibrate the proposal probabilities and repeat the procedure until a signal is released, so that the resulting experiment is exactly $H$. Third, we bound the expected sampling cost in terms of the evidence carried by $H$. Finally, we apply the construction to the full record of a stopping policy for $G$ to obtain \eqref{eq:uniform-conversion}. The formal proof is in Appendix~\ref{app:conversion-record}. Readers interested mainly in the implications can skip to \Cref{sec:costs}.

\paragraph{Step 1: Propose and verify.}
We first consider reproducing a target experiment $H$. The challenge is to match its state-contingent signal distributions exactly without knowing the true state. To address this, we use independent randomization to propose a state $k$ and a signal realization $z$. The state $k$ serves only as a guess, while $z$ is the signal we may release. We then sample $F$ sequentially and accumulate likelihood-ratio evidence. We release $z$ if, for every alternative state $j$, the evidence for $k$ against $j$ reaches $\log(h_k(z)/h_j(z))+\lambda$. The first term represents the evidence carried by $z$, while $\lambda>0$ makes acceptance under an incorrect proposed state unlikely. Alternatively, if the evidence against $k$ reaches $\lambda$ for any alternative, we discard the proposal without releasing a signal. Thus, the acceptance rule makes proposals based on an incorrect state unlikely to be accepted.

\paragraph{Step 2: Calibrate and repeat.}
We choose the probabilities of proposing $(k,z)$ so that, in every state $i$, the procedure releases a signal with probability $q\in(0,1)$ independent of the state, and the released signal follows distribution $H_i$. Such probabilities exist because proposals based on the true state are accepted with high probability, whereas acceptance under an incorrect proposed state is bounded by the likelihood ratios of $z$. These properties ensure that the system of linear equations determining the proposal probabilities admits a positive solution. Consequently, the procedure releases no signal with probability $1-q$ in every state, so rejection reveals no information about the state. Repeating the procedure with fresh observations until a signal is released therefore reproduces $H$ exactly.

\paragraph{Step 3: Bound the sampling cost.}
Suppose the true state is $i$. If the proposed state $k$ is incorrect, the evidence against $k$ has positive drift, so the procedure ends after a bounded expected number of observations. If $k=i$ but the proposal is rejected, a change-of-measure argument shows that the expected sampling cost contributed by these rare events is also bounded. The main cost comes from releasing a signal $z$ when $k=i$. In this case, $F$ must accumulate sufficient evidence against every alternative, requiring approximately
\[
b_i(z):=\max_{j\ne i}\frac{[\log(h_i(z)/h_j(z))]^+}{I_{ij}(F)}
\]
observations, up to a square-root correction arising from fluctuations of the likelihood ratios. Let $B_i:=\int b_i(z)\,H_i(dz)$ denote the average evidence requirement. By the calibration in Step 2, the probability of proposing the correct state and releasing a signal in $dz$ is at most $qH_i(dz)$. Thus, the leading expected sampling cost of one round is at most $qB_i$, with a remainder bounded by $qA_F(1+\sqrt{B_i})$. Since the expected number of rounds until a signal is released is $1/q$, the factor $q$ cancels, yielding
\[
\E_i^F N\le B_i+A_F(1+\sqrt{B_i}),
\]
where $A_F$ depends only on $F$, not on the target experiment $H$.

\paragraph{Step 4: Apply the construction to stopping policies.}
Consider any policy $(T,Z)$ for $G$. We apply Steps 1--3 to the experiment generated by its full record $R=(T,X_1,\ldots,X_T)$, and then use independent randomization to release $Z$. Since the log-likelihood ratios of the record equal the accumulated evidence from $G$, its average evidence requirement is
\[
B_i=\E_i^G\left[
\max_{j\ne i}\frac{[S_{ij}^G(T)]^+}{I_{ij}(F)}
\right]
\le \rho_i(F\to G)\,\E_i^GT
     +A_{F,G}\sqrt{\E_i^GT}.
\]
The leading term follows from the KL rates of $G$ relative to $F$, while the square-root remainder accounts for fluctuations in the accumulated evidence. Combining this with the cost bound from Step 3 yields \eqref{eq:uniform-conversion}.

\section{The Cost of Producing Information}\label{sec:costs}

The preceding sections use exact reproduction to compare experiments by the information they can generate. We now fix the source experiment and study the cost of producing information through sequential sampling. Consider an information provider, such as a product-testing laboratory, that samples a source experiment at a fixed cost per observation to produce a target experiment $H$. The lower bound \eqref{eq:product-lower} and \Cref{thm:conversion} together yield an explicit formula for the minimum expected cost per copy as the scale of production grows. We then characterize the resulting cost function over experiments and examine its properties.

\subsection{Information production}

We derive information costs from the minimum expenditure required to produce a target experiment through sequential sampling.\footnote{Related approaches to deriving information costs from dynamic information acquisition include \citet{MorrisStrack2019,BloedelZhong2025}.} 
Fix a source experiment $E$ satisfying \Cref{ass:regular}, a cost $c>0$ per observation, and a prior $\mu$. To produce a target experiment $H$, the producer sequentially samples $E$, chooses when to stop, and uses the observations to generate a signal with the distribution specified by $H$ in each state. The minimum expected cost of producing one copy of $H$ is
\[
C^E(H):=\inf\Bigl\{c\sum_{i\in\Theta}\mu_i\E_i^E T:
(T,Z)\text{ is a policy for $E$ with }H(T,Z)=H\Bigr\}.
\]
Throughout this section, we assume that target experiments $H$ satisfy \Cref{ass:regular}, except that distinct states may induce the same signal distribution.

Producing multiple conditionally independent copies of $H$ jointly can be cheaper than producing each copy separately. Because information accumulates randomly, producing a finite batch may require additional sampling beyond its average information requirements. Joint production allows the producer to spread this additional cost across the batch. In particular, producing two batches separately with fresh observations is always feasible, so
\[
C^E(H^{\otimes(m+m')})\le C^E(H^{\otimes m})+C^E(H^{\otimes m'}).
\]
Thus, production costs are subadditive in the batch size, and the minimum expected cost per copy has a well-defined limit:
\begin{equation}\label{eq:cost-function}
C_\infty^E(H):=\lim_{m\to\infty}\frac{C^E(H^{\otimes m})}{m}
=\inf_{m\ge1}\frac{C^E(H^{\otimes m})}{m}.
\end{equation}
We focus on this \emph{large-scale unit cost}, which captures the first-order cost of information production when the producer can exploit joint production. Holding the source experiment fixed, it defines a cost function over target experiments induced by the underlying sampling technology.

\begin{theorem}[Large-scale unit cost]\label{thm:costs}
Fix a cost $c>0$ per observation.
\begin{enumerate}[label=\textup{(\alph*)}]
 \item For every prior $\mu$, source experiment $E$, and target experiment $H$,
 \begin{equation}\label{eq:per-copy-formula}
  C_\infty^E(H)=c\sum_{i\in\Theta}\mu_i\,\rho_i(E\to H)=c\sum_{i\in\Theta}\mu_i\max_{j\ne i}\frac{I_{ij}(H)}{I_{ij}(E)}.
 \end{equation}
 \item Fix a full-support prior $\mu$. A function $\Phi$ on target experiments equals $C_\infty^E$ for some source experiment $E$ if and only if there are positive weights $(\pi_{ij})_{i\ne j}$ such that
 \[
  \Phi(H)=\sum_{i\in\Theta}\max_{j\ne i}\pi_{ij}I_{ij}(H)\qquad\text{for every target experiment }H.
 \]
Moreover, if $\Phi$ and $\widetilde{\Phi}$ have weights $\pi$ and $\widetilde{\pi}$, respectively, then $\Phi(H)\le\widetilde{\Phi}(H)$ for every $H$ implies $\pi_{ij}\le\widetilde{\pi}_{ij}$ for every $i\ne j$. In particular, the representing weights are unique.
\end{enumerate}
\end{theorem}

To interpret part (a), note that in state $i$, each copy of $H$ requires $I_{ij}(H)$ units of expected evidence against every alternative state $j$, while each observation of $E$ supplies $I_{ij}(E)$ units on average. Since each observation provides evidence against all alternatives simultaneously, the number of observations required is determined by the most demanding comparison, or the \emph{bottleneck}. This gives a lower bound on production costs: by \eqref{eq:product-lower}, producing $m$ copies of $H$ requires at least $m\rho_i(E\to H)$ observations in expectation in state $i$. Theorem~\ref{thm:conversion} shows that this bound is asymptotically attainable. Applying the theorem to a policy that samples $H$ exactly $m$ times yields a policy for $E$ that reproduces $H^{\otimes m}$ exactly, using at most $m\rho_i(E\to H)+A_{E,H}\sqrt m$ observations in expectation in every state. The additional sampling cost vanishes per copy as $m$ grows, establishing the formula in part (a).

Part (b) characterizes the information-cost functions that can arise from sequential sampling. The necessity of the stated form follows immediately from part (a), with weights $\pi_{ij}=c\mu_i/I_{ij}(E)$. Each weight represents the cost of acquiring one unit of expected evidence against state $j$ in state $i$, weighted by the prior probability of state $i$.

The converse of part (b) rests on a key observation: every strictly positive array of pairwise KL divergences can be generated by some experiment. This is not immediate since all the divergences must arise from the same collection of state-dependent signal distributions. Indeed, when the signal space is small, not every array is attainable.\footnote{For example, with three or more states, a binary-signal experiment has only $|\Theta|$ free parameters, whereas there are $|\Theta|(|\Theta|-1)$ directed KL divergences.}
Nevertheless, a sufficiently rich but finite signal space allows us to realize any positive array. The idea is to construct, for each ordered pair $(i,j)$, a binary experiment that generates much more expected evidence in state $i$ against state $j$ than for any other ordered pair. We then mix these binary experiments, revealing which component was selected. Since the KL divergences of the mixture are weighted sums of those of its components, appropriately chosen mixture weights, with any remaining probability assigned to an uninformative experiment, allow us to match any prescribed positive array exactly.

With this observation, we can construct a source experiment $E$ with $I_{ij}(E)=c\mu_i/\pi_{ij}$ for every ordered pair. By part (a), this source generates exactly the desired cost function.
The ordering and uniqueness of the weights follow from the same observation. To see this, fix an ordered pair $(i,j)$. For any $\varepsilon>0$, we can construct a target experiment $H^\varepsilon$ with $I_{ij}(H^\varepsilon)=1$ and all other pairwise KL divergences equal to $\varepsilon$. As $\varepsilon$ approaches zero, the representation in part (b) gives $\lim_{\varepsilon\downarrow0}\Phi(H^\varepsilon)=\pi_{ij}$.
Thus, pointwise ordering of cost functions implies componentwise ordering of their weights. In particular, equality of cost functions implies uniqueness of the representing weights.

Theorem~\ref{thm:costs} also connects the cost of producing information to our earlier ranking of experiments.
\begin{corollary}[Cost ranking]\label{cor:cost-ranking}
Fix a cost $c>0$ and a full-support prior $\mu$. Let $F$ and $G$ be non-tied experiments. Then $F$ large-budget dominates $G$ if and only if $C_\infty^F(H)\le C_\infty^G(H)$ for every target experiment $H$.
\end{corollary}
The forward implication follows from part (a) and Theorem~\ref{thm:generic-equivalence}: a source with higher pairwise KL divergences can produce every target experiment at a lower large-scale unit cost. Conversely, the ordering property in part (b) implies that a source that is no more costly for every target must have weakly higher pairwise KL divergences. For non-tied experiments, Theorem~\ref{thm:generic-equivalence} then yields large-budget dominance.

\subsection{Properties of large-scale unit costs}

The large-scale unit cost satisfies several standard properties. First, it is monotone in Blackwell's order: garbling a target experiment cannot increase its production cost, since garbling weakly reduces every pairwise KL divergence. Second, it has constant returns to replication, $C_\infty^E(H^{\otimes m})=mC_\infty^E(H)$. Third, it is linear under dilution: if $H^\alpha$ delivers a signal from $H$ with probability $\alpha$ and an uninformative signal otherwise, then $C_\infty^E(H^\alpha)=\alpha C_\infty^E(H)$.

With two states, the cost function is also additive. Each state has only one alternative, so the maximum in \eqref{eq:per-copy-formula} disappears and the cost reduces to a weighted sum of pairwise KL divergences. Since KL divergences add across conditionally independent signals, producing two targets jointly costs exactly as much as producing them separately.\footnote{This is the additive KL form characterized by \citet{PomattoStrackTamuz2023}; related foundations based on dynamic information acquisition are provided by \citet{MorrisStrack2019,BloedelZhong2025}.}

With three or more states, joint production can instead be strictly cheaper. For target experiments $H$ and $H'$, let $H\otimes H'$ denote the experiment that delivers one signal from each, independently conditional on the state. Since $I_{ij}(H\otimes H')=I_{ij}(H)+I_{ij}(H')$, we have
\[
\rho_i(E\to H\otimes H')
=\max_{j\ne i}\frac{I_{ij}(H)+I_{ij}(H')}{I_{ij}(E)}
\le \rho_i(E\to H)+\rho_i(E\to H').
\]
Thus, the cost is subadditive:
\[
C_\infty^E(H\otimes H')\le C_\infty^E(H)+C_\infty^E(H').
\]
The inequality is strict if and only if the two target experiments have no common bottleneck alternative in some state with positive prior probability. With three or more states, we can always construct two targets with distinct bottlenecks in such a state. Thus, every cost function in Theorem~\ref{thm:costs} exhibits economies of scope for some pair of targets, in the sense of \citet{PanzarWillig1981}, and is therefore nonadditive (see Appendix~\ref{app:scope}).

\begin{example}\label{ex:scope}
A product is sound (state $0$) or has one of two defects (states $1$ and $2$). A binary measurement $E$ reads high with probabilities $1/2$, $p$, and $1-p$ in states $0$, $1$, and $2$, respectively, where $1/2<p<1$. Two target experiments each flag one defect: $H$ reads high with probability $p$ under the first defect and $1/2$ otherwise, while $H'$ reads high with probability $p$ under the second defect and $1/2$ otherwise.

Let $d$ denote the KL divergence of a fair coin from a coin with bias $p$. The following table reports the relevant KL divergences and the large-scale sampling cost per copy in state $0$:

\begin{center}
\begin{tabular}{@{}lccc@{}}
\toprule
 & $I_{01}$ & $I_{02}$ & State-$0$ unit cost \\
\midrule
Source $E$ & $d$ & $d$ & --- \\
Target $H$ & $d$ & $0$ & $c$ \\
Target $H'$ & $0$ & $d$ & $c$ \\
Joint target $H\otimes H'$ & $d$ & $d$ & $c$ \\
\bottomrule
\end{tabular}
\end{center}

In state $0$, each target separately requires one observation of $E$ per copy, since $\rho_0(E\to H)=\rho_0(E\to H')=1$. Yet their joint production also requires only one observation per copy of the pair. Thus, joint production reduces the sampling cost in state $0$ from $2c$ to $c$. Since joint production never increases the required sampling in any state, this yields strict economies of scope under any full-support prior. 

Economies of scope arise because separate production may leave some of the evidence supplied by source observations unused, whereas joint production can make use of it. In the example, when the true state is $0$, each observation of $E$ provides evidence against both states $1$ and $2$. Producing $H$ alone requires evidence only against state $1$, leaving the evidence against state $2$ unused; the reverse holds for $H'$. Joint production allows the producer to use the evidence against both alternatives from the same observations, reducing the total sampling requirement. Thus, economies of scope arise from sharing information inputs across different targets.
\end{example}

Finally, our large-scale unit costs connect to the axiomatic framework of \citet{BordoliIijima2025}, who characterize the Max-KL costs of the form $\max_{\beta\in\Pi}\sum_{i\ne j}\beta_{ij}I_{ij}(H)$, where $\Pi$ is a compact set of arrays of nonnegative weights. Our sequential-sampling model provides a foundation for a particular subclass of these costs, characterized by the sum of statewise maxima in Theorem~\ref{thm:costs}(b).\footnote{For each selection $s$ with $s(i)\ne i$ for every state $i$, define $\beta^s_{ij}=\pi_{ij}\mathbf{1}\{j=s(i)\}$. We can take $\Pi$ to be the set of all such arrays $\beta^s$, so that maximizing over $\Pi$ gives the sum of statewise maxima.}
Our characterization thus identifies exactly which Max-KL costs arise from sequentially sampling a fixed experiment. Moreover, the weights have a direct economic interpretation: $\pi_{ij}=c\mu_i/I_{ij}(E)$ is the prior-weighted cost of acquiring one unit of expected evidence against state $j$ in state $i$. This restriction is substantive. With two states, our costs are additive, whereas general Max-KL costs need not be. With three or more states, our costs necessarily exhibit economies of scope for some pair of targets, whereas the broader Max-KL class also includes additive costs.

\section{Concluding Remarks}\label{sec:conclusion}

This paper studies the ranking of experiments and the cost of producing information under sequential sampling. Endogenous stopping allows a decision maker to continue sampling when the evidence remains inconclusive, rather than committing to a fixed sample size. We show that, for non-tied pairs of experiments, large-budget and small-cost dominance coincide and are characterized by KL dominance. The resulting ranking is thus governed by the expected evidence each observation provides for each state against every alternative. This criterion differs from the Chernoff-information and R\'enyi-divergence criteria for fixed-sample comparisons in \citet{MoscariniSmith2002,MuPomattoStrackTamuz2021}. Indeed, allowing endogenous stopping can reverse the ranking of experiments.

The sufficiency of KL dominance rests on our exact reproduction theorem. Any terminal experiment generated by a stopping policy based on one source can be reproduced exactly from another, with statewise expected sample sizes bounded by sharp first-order KL-ratio conversion rates and a square-root remainder. Crucially, the remainder is uniform over the stopping policy being reproduced. Together with the KL conversion rates, this uniformity establishes production dominance: under KL dominance, there is a common budget threshold beyond which one source can reproduce every terminal experiment attainable from the other within the same budget, for every prior. This gives the cycle of implications in \Cref{fig:flow}.

The same reproduction result also characterizes the cost of producing information at scale. For a fixed source, the large-scale unit cost equals the sampling cost times a prior-weighted sum of statewise KL conversion rates. This identifies exactly which of the Max-KL costs characterized by \citet{BordoliIijima2025} arise from sequential sampling of a fixed experiment. Finally, we show that these costs are additive with two states but exhibit economies of scope with three or more states. Such cost reductions arise when different targets have different bottleneck alternatives, allowing joint production to exploit evidence that would otherwise go unused under separate production.

Several questions remain for future research. First, both of our dominance orders allow the decision maker to allocate observations across histories, subject only to an expected-sample budget or a constant cost per observation. Hard caps on the number of observations or convex sampling costs restrict this flexibility and may lead to different rankings of experiments. Second, allowing the decision maker to choose adaptively among several source experiments, as in \citet{Chernoff1959,NitinawaratAtiaVeeravalli2013,LiangMuSyrgkanis2022}, raises the further question of how sets of experiments should be ranked. Finally, the restriction to non-tied pairs is substantive. When the directed KL divergences of two sources coincide for some ordered pairs of states, the joint fluctuations of evidence against different alternatives can matter beyond their average rates, and weak KL dominance need not imply large-budget dominance. Characterizing dominance in the presence of ties remains an open question.

\clearpage
\bibliographystyle{ecta}
\bibliography{KL}

\clearpage

\part*{Appendix}
\addcontentsline{toc}{part}{Appendix}
\appendix
\numberwithin{equation}{section}
\numberwithin{lemma}{section}

\section{Preliminaries}\label{app:assumptions}

Throughout the appendices, signal spaces are standard Borel spaces, and the experiments that are sampled satisfy Assumption~\ref{ass:regular}. For probability distributions $P$ and $Q$, $\KL{P}{Q}$ denotes the KL divergence of $P$ from $Q$, so that $I_{ij}(E)=\KL{E_i}{E_j}$. For a sampled experiment $E$, write
$S^E_{ij}(n):=\sum_{t=1}^n
\log\frac{dE_i}{dE_j}(X_t)$
and $M^E_{ij}(n):=S^E_{ij}(n)-nI_{ij}(E)$. The randomization of a policy is included in the information with respect to which its stopping time is defined; because the randomization is independent of the state and of the observations, each observation is independent of this information at earlier dates, conditional on the state.

For a decision problem $\mathcal D=(\mu,\mathcal A,u)$,
write $\ell(a,i):=\max_b u(b,i)-u(a,i)$ for the loss
of action $a$ in state $i$. The value of information
and the maximum statewise loss are, respectively,
\[
v(\mathcal D):=\min_a\sum_i\mu_i\ell(a,i),
\qquad
\ell^*(\mathcal D):=
\max_{\substack{i:\mu_i>0\\a\in\mathcal A}}\ell(a,i).
\]

\begin{lemma}[Stopped likelihood ratios]\label{lem:likelihood}
Let $(T,Z)$ be a policy for an experiment $E$, and let $i$ and $j$ be states. The distributions in states $i$ and $j$ of the observations up to time $T$, together with the randomization, are mutually absolutely continuous, with log-likelihood ratio $S^E_{ij}(T)$. Moreover,
\begin{align}
 \E_i^E\bigl[S^E_{ij}(T)\bigr]&=I_{ij}(E)\,\E_i^ET,\label{eq:wald-app}\\
 \E_i^E\bigl[M^E_{ij}(T)^2\bigr]&=\operatorname{Var}_i\Bigl(\log\frac{dE_i}{dE_j}(X_1)\Bigr)\E_i^ET,\label{eq:stopped-square}
\end{align}
and $I_{ij}(H(T,Z))\le I_{ij}(E)\,\E_i^ET$.
\end{lemma}

\begin{proof}
At any fixed date, the probability of stopping and issuing a signal, conditional on the observations, is specified by the policy's randomization and does not depend on the state. That factor cancels from the likelihood ratio. Summing over the countably many finite stopping dates proves the identity for the likelihood ratio and the mutual absolute continuity. The same argument applies to a signal in a standard Borel space through its conditional distribution given the observations. For any integrable one-observation increment $Y_t$, independence from the preceding history gives
\[
 \E_i\sum_{t\le T}|Y_t|=\E_i|Y_1|\,\E_iT<\infty.
\]
Applying this identity to the log-likelihood increments justifies the summation and proves \eqref{eq:wald-app}. For the centered increments, whose variance we denote by $\sigma^2_{ij}$, orthogonality gives
\[
 \E_i\bigl|M^E_{ij}(T\wedge \ell)-M^E_{ij}(T\wedge m)\bigr|^2
 =\sigma^2_{ij}\,\E_i(T\wedge\ell-T\wedge m)\qquad(\ell\ge m).
\]
Since $\E_iT<\infty$, these stopped sums are Cauchy in $L^2$ and converge almost surely to $M^E_{ij}(T)$. Taking limits in the corresponding identity for bounded times proves \eqref{eq:stopped-square}. Finally, the signal is obtained from the observations and the randomization by a transformation that does not depend on the state, so the data-processing inequality for KL divergences gives the last claim.
\end{proof}

\begin{lemma}[Simultaneous evidence requirements]\label{lem:passage}
Let $S_1(n),\ldots,S_L(n)$ be random walks whose increment vectors are independent and identically distributed, with positive means $d_1,\ldots,d_L$ and finite second moments. For real levels $c_1,\ldots,c_L$, let
\[
 \tau:=\inf\{n\ge0:S_\ell(n)\ge c_\ell\text{ for every }\ell\},
 \qquad t_*:=\max_\ell c_\ell^+/d_\ell.
\]
Then, for a constant $C$ that depends only on the distribution of the increments,
\begin{equation}\label{eq:passage}
 \E\tau\le t_*+C\sqrt{t_*}+C.
\end{equation}
\end{lemma}

\begin{proof}
Write $M_\ell(n):=S_\ell(n)-d_\ell n$ and $\tau_m:=\tau\wedge m$. If $1\le\tau\le m$, at least one coordinate was below its level at $\tau-1$, and at $\tau$ its value is at most that level plus its last increment. If $\tau>m$, at least one coordinate is still below its level at $m$. Consequently,
\begin{equation}\label{eq:passage-path}
 \tau_m\le t_*+\sum_\ell\frac{|M_\ell(\tau_m)|+|\Delta S_\ell(\tau_m)|}{d_\ell},
\end{equation}
where $\Delta S_\ell(t)$ is the $t$-th increment and $\Delta S_\ell(0):=0$. The last-increment term is not needed when $\tau>m$, but including it yields a single inequality.

For the bounded stopping time $\tau_m$, the independence of each increment from the preceding history gives
\[
 \E M_\ell(\tau_m)^2=\operatorname{Var}(\Delta S_\ell(1))\,\E\tau_m,
 \qquad
 \E\sum_{t\le\tau_m}|\Delta S_\ell(t)|^2=\E|\Delta S_\ell(1)|^2\,\E\tau_m.
\]
The second identity bounds the last increment, since its square is at most the sum on the left. Taking expectations in \eqref{eq:passage-path} and applying the Cauchy--Schwarz inequality therefore yields
\[
 \E\tau_m\le t_*+C\sqrt{\E\tau_m}.
\]
Here and below we solve such inequalities with the elementary implication
\begin{equation}\label{eq:quadratic-step}
 x\le a+C\sqrt x\ \Longrightarrow\ x\le a+C\sqrt a+C^2
 \qquad(0\le x<\infty,\ a,C\ge0),
\end{equation}
obtained by solving the quadratic inequality in $\sqrt x$ and using $\sqrt{C^2+4a}\le C+2\sqrt a$. With $a=t_*$, it gives \eqref{eq:passage} for $\tau_m$, after enlarging $C$, uniformly in $m$, and monotone convergence proves that $\tau$ is integrable and satisfies the same bound.
\end{proof}

Two further elementary facts are used repeatedly. For $x\in[0,1]$ and $y\in(0,1)$,
\begin{equation}\label{eq:bernoulli-bound}
 \KL{\operatorname{Ber}(x)}{\operatorname{Ber}(y)}\ge x\log(1/y)-\log2,
\end{equation}
because the left side equals $x\log(1/y)+(1-x)\log(1/(1-y))-h(x)$, where $h(x)\le\log2$ is the entropy of $\operatorname{Ber}(x)$ and the middle term is nonnegative.

\begin{lemma}[Identifying the state from a fixed sample]\label{lem:ml}
Let $E$ be an experiment whose distributions are distinct and mutually absolutely continuous, and let $\hat\theta_m$ be a state that maximizes the likelihood of the first $m$ observations, with ties broken by a fixed order. Then
\[
 \Prob_i^E(\hat\theta_m\ne i)\le(K-1)\,\eta_E^m\quad\text{for every state }i,\qquad\text{where}\quad\eta_E:=\max_{i\ne j}\int\sqrt{dE_i\,dE_j}<1.
\]
Consequently, for every finite decision problem $\mathcal D$, every $c>0$, and every integer $m\ge0$,
\[
 W_c(E;\mathcal D)\ge V^*(\mathcal D)-(K-1)\,\ell^*(\mathcal D)\,\eta_E^m-cm.
\]
\end{lemma}

\begin{proof}
If $\hat\theta_m\ne i$, then $S^E_{ji}(m)\ge0$ for some $j\ne i$. By Markov's inequality,
\[
 \Prob_i^E\bigl(S^E_{ji}(m)\ge0\bigr)\le\E_i^E e^{S^E_{ji}(m)/2}=\Bigl(\int\sqrt{dE_i\,dE_j}\Bigr)^m,
\]
and the integral is below one because $E_i\ne E_j$. For the second claim, consider the policy that takes $m$ observations and then the action that is optimal in state $\hat\theta_m$. In a state $i$ with $\mu_i>0$, its loss is zero when $\hat\theta_m=i$ and at most $\ell^*(\mathcal D)$ otherwise.
\end{proof}

For experiments $E^1,\ldots,E^L$ with mutually absolutely continuous distributions, possibly with different signal spaces and with distributions that may coincide across states, and positive weights $w_1,\ldots,w_L$ that sum to one, the \emph{labelled mixture} $\sum_\ell w_\ell E^\ell$ delivers, at each observation, an observation of $E^\ell$ together with the label $\ell$ with probability $w_\ell$. The $\alpha$-dilution $E^\alpha$ of Section~\ref{subsec:fixed} is the labelled mixture of $E$, with weight $\alpha$, and an uninformative experiment. The product $E^1\otimes E^2$ delivers one observation of each experiment, independent conditional on the state.

\begin{lemma}[Mixtures and products]\label{lem:mixing}
For all distinct states $i$ and $j$,
\[
 I_{ij}\Bigl(\sum_\ell w_\ell E^\ell\Bigr)=\sum_\ell w_\ell I_{ij}(E^\ell),\qquad
 I_{ij}(E^1\otimes E^2)=I_{ij}(E^1)+I_{ij}(E^2).
\]
If $\Theta=\{0,1\}$, then, for every $s\in[0,1]$,
\[
 e^{\psi_{\sum_\ell w_\ell E^\ell}(s)}=\sum_\ell w_\ell\,e^{\psi_{E^\ell}(s)},\qquad
 \psi_{E^1\otimes E^2}(s)=\psi_{E^1}(s)+\psi_{E^2}(s).
\]
\end{lemma}

\begin{proof}
On the signals with label $\ell$, the distribution in state $i$ of the labelled mixture is $w_\ell E^\ell_i$. The weight cancels from the likelihood ratio, which gives the identities for mixtures block by block. For products, log-likelihood ratios add and the integrands in the definition of $\psi$ multiply.
\end{proof}

\section{Proof of Theorem~\ref{thm:conversion}}\label{app:conversion-record}

We first take a target experiment $H=(H_i)_{i\in\Theta}$ as given. Its signal space $\mathcal Z$ may be any standard Borel space, and its distributions are mutually absolutely continuous; they may coincide. Let $h_i$ be their densities with respect to $K^{-1}\sum_jH_j$, chosen positive outside a common null set, and define
\begin{equation}\label{eq:b}
 b_i(z):=\max_{j\ne i}\frac{[\log(h_i(z)/h_j(z))]^+}{I_{ij}(F)}.
\end{equation}
Assume that $\int b_i\,dH_i<\infty$ for every $i$. As in Section~\ref{subsec:construction}, $b_i(z)$ is the largest, over alternatives, of the evidence that $z$ carries divided by its mean rate under $F$; the maximum is taken before averaging over signal values. Fix the margin $\lambda:=\log(4K)$.\footnote{Any margin $\log(\gamma K)$ with $\gamma>2$ would serve. The margin ensures the diagonal dominance in \eqref{eq:coefficients}; the choice $\gamma=4$ gives the bounds in \eqref{eq:a-bounds}.}

Throughout this appendix, probabilities and expectations without a superscript refer to $F$, and $S_{ij}:=S^F_{ij}$. Constants denoted $A_F$ or $A_{F,G}$ may increase from one display to the next and depend only on the indicated experiments.

\subsection{One attempt}

For each candidate state $k$ and candidate signal $z$, define
\begin{align}
 \tau^+(k,z)&:=\inf\Bigl\{n\ge0:S_{kj}(n)\ge\log\frac{h_k(z)}{h_j(z)}+\lambda\text{ for every }j\ne k\Bigr\},\label{eq:plus}\\
 \tau^-(k)&:=\inf\bigl\{n\ge0:S_{jk}(n)\ge\lambda\text{ for some }j\ne k\bigr\},\label{eq:minus}\\
 \sigma(k,z)&:=\tau^+(k,z)\wedge\tau^-(k).\label{eq:sigma}
\end{align}
At time $\sigma(k,z)$, the attempt issues $z$ if $\tau^+(k,z)<\tau^-(k)$; otherwise it returns the symbol $\partial\notin\mathcal Z$, which means that the attempt is abandoned. Thus a tie abandons. Since $\lambda>0$, $\tau^-(k)\ge1$. The signal is issued at time zero exactly when $h_j(z)/h_k(z)\ge e^\lambda$ for every $j\ne k$. The levels in \eqref{eq:plus} depend on the likelihood ratios of the candidate signal, whereas the level in \eqref{eq:minus} is fixed, so issuing a signal can require a long sample, whereas abandoning requires only a fixed amount of evidence against $k$.

Let
\[
 p_i(k,z):=\Prob_i\bigl(\tau^+(k,z)<\tau^-(k)\bigr).
\]
For each $z$, choose coefficients $w_1(z),\ldots,w_K(z)$ that solve
\begin{equation}\label{eq:coefficients}
 \sum_{k\in\Theta}h_k(z)\,p_i(k,z)\,w_k(z)=h_i(z)\qquad\text{for every state }i.
\end{equation}
Lemma~\ref{lem:law} shows that this system has a unique solution, which is positive and measurable in $z$. Set
\begin{equation}\label{eq:q}
 q:=\Bigl(\sum_{k\in\Theta}\int w_k(z)\,H_k(dz)\Bigr)^{-1}.
\end{equation}
An attempt draws a candidate $(k,z)$ from the distribution $q\,w_k(z)\,H_k(dz)$, independently of the observations, and then applies \eqref{eq:plus}--\eqref{eq:sigma} and the rule above. Denote its duration by $\nu$ and its output, in $\mathcal Z\cup\{\partial\}$, by $O$. The candidate state is drawn, never observed; the probabilities used in the draw are determined by the known distributions of $F$ and $H$. The same signal can be issued for different candidate states, and \eqref{eq:coefficients} chooses their frequencies so that these contributions add up to exactly $q\,h_i(z)$ in every state.

\begin{lemma}[Signal distribution in one attempt]\label{lem:law}
For each $z$, the system \eqref{eq:coefficients} has a unique solution in $\R^K$; it is positive, measurable in $z$, and satisfies
\begin{equation}\label{eq:a-bounds}
 \frac23\le w_k(z)\le\frac43,\qquad w_i(z)\,p_i(i,z)\le1.
\end{equation}
Moreover, $0<q<1$, every $\sigma(k,z)$ has finite mean in every state, and, for every measurable $A\subseteq\mathcal Z$,
\begin{equation}\label{eq:fraction}
 \Prob_i(O\in A)=q\,H_i(A),\qquad\Prob_i(O=\partial)=1-q.
\end{equation}
\end{lemma}

\begin{lemma}[Expected duration of one attempt]\label{lem:cost}
For a constant $A_F$ that does not depend on $H$,
\begin{equation}\label{eq:cost}
 \E_i\nu\le q\Bigl\{\int b_i(z)\,H_i(dz)+A_F\Bigl(1+\int\sqrt{b_i(z)}\,H_i(dz)\Bigr)\Bigr\}\qquad\text{for every state }i.
\end{equation}
\end{lemma}

Both lemmas carry the factor $q$: an attempt delivers the fraction $q$ of the target distribution in every state, and its expected duration is at most $q$ times the bound on the total.

\subsection{Restarting reproduces \texorpdfstring{$H$}{H} exactly}

Repeat attempts, with fresh observations and independent randomization each time, until one issues a signal in $\mathcal Z$. If $(\nu_r,O_r)$ are the duration and output of attempt $r$, define
\[
 J:=\inf\{r\ge1:O_r\ne\partial\},\qquad N:=\sum_{r=1}^J\nu_r,\qquad Y:=O_J.
\]
Conditional on each state, the attempts are independent and identically distributed. Lemma~\ref{lem:law} gives $\Prob_i(J>r)=(1-q)^r$, so $J$ is finite almost surely. For every measurable $A$,
\begin{equation}\label{eq:exact}
 \Prob_i(Y\in A)=\sum_{r=0}^\infty(1-q)^r\,q\,H_i(A)=H_i(A).
\end{equation}
Likewise, Tonelli's theorem and Lemma~\ref{lem:cost} give
\begin{align}
 \E_iN&=\sum_{r=0}^\infty(1-q)^r\,\E_i\nu=q^{-1}\E_i\nu\nonumber\\
 &\le\int b_i(z)\,H_i(dz)+A_F\Bigl(1+\int\sqrt{b_i(z)}\,H_i(dz)\Bigr).\label{eq:H-cost}
\end{align}
For the first equality, the event $\{J>r\}$ depends only on the first $r$ attempts and is independent of $\nu_{r+1}$; the duration and the output of a single attempt may be correlated. The resulting $N$ is a stopping time with finite mean in every state. An attempt may use no observations, but there are only finitely many attempts almost surely. Equations~\eqref{eq:exact} and \eqref{eq:H-cost} prove exact reproduction of the given target experiment $H$. By Jensen's inequality, \eqref{eq:H-cost} also gives the form stated in Section~\ref{subsec:construction},
\[
 \E_iN\le\int b_i\,dH_i+A_F\Bigl(1+\sqrt{\textstyle\int b_i\,dH_i}\Bigr).
\]

\subsection{Applying the construction to the record}

Suppose first that the policy $(T,Z)$ for $G$ satisfies $T\ge1$. Its record
\[
 R:=(T,X_1,\ldots,X_T,Z),\qquad H_i:=\Law_i^G(R),
\]
takes values in a countable disjoint union of finite products of standard Borel spaces. By Lemma~\ref{lem:likelihood},
\begin{equation}\label{eq:record-LR}
 \log\frac{dH_i}{dH_j}(R)=S^G_{ij}(T)=I_{ij}(G)\,T+M^G_{ij}(T),
\end{equation}
where $M^G_{ij}$ has centered, independent, and identically distributed increments. In particular, the distributions of the record are mutually absolutely continuous. Write $t_i:=\E_i^GT$. Identity \eqref{eq:stopped-square} and the Cauchy--Schwarz inequality imply
\[
 \E_i^G\bigl|M^G_{ij}(T)\bigr|\le\Bigl(\operatorname{Var}_i^G\Bigl(\log\frac{dG_i}{dG_j}(X_1)\Bigr)\Bigr)^{1/2}\sqrt{t_i}.
\]
By the definitions of $b_i$ and $\rho_i:=\rho_i(F\to G)$,
\begin{equation}\label{eq:pointwise-record}
 b_i(R)\le\rho_i\,T+\sum_{j\ne i}\frac{|M^G_{ij}(T)|}{I_{ij}(F)}.
\end{equation}
Consequently,
\[
 \int b_i\,dH_i\le\rho_it_i+A_{F,G}\sqrt{t_i},\qquad\int\sqrt{b_i}\,dH_i\le\sqrt{\rho_it_i+A_{F,G}\sqrt{t_i}},
\]
where the second inequality is Jensen's. In particular, the requirements are integrable, so the construction applies. Because $t_i\ge1$, substituting these bounds into \eqref{eq:H-cost} gives
\[
 \E_i^FN\le\rho_it_i+A_{F,G}\sqrt{t_i}.
\]
Reproducing $R$ by the construction and releasing only its last coordinate proves Theorem~\ref{thm:conversion} for policies with $T\ge1$, subject to Lemmas~\ref{lem:law} and~\ref{lem:cost}, which are proved below. No truncation of $T$ is needed: its integrability is exactly what \eqref{eq:stopped-square} requires. Nothing in the construction requires the distributions of the target to differ across states, which is why Theorem~\ref{thm:conversion} allows $G$ to have the same distribution in distinct states.

If the policy for $G$ can stop at time zero, copy its initial randomization. Let $\alpha$ be the probability that it takes an observation; this probability does not depend on the state. If $\alpha=0$, the terminal experiment requires no observations. Otherwise, reproduce the signal issued at time zero without sampling, and apply the preceding construction conditional on taking an observation. The conditional distributions of the record are still mutually absolutely continuous with log-likelihood ratio $S^G_{ij}(T)$, because the state-independent probability $\alpha$ cancels from the ratio. The conditional mean duration is $t_i/\alpha\ge1$, and multiplying its bound by $\alpha$ gives
\[
 \E_i^FN\le\rho_it_i+A_{F,G}\sqrt{\alpha t_i}\le\rho_it_i+A_{F,G}\sqrt{t_i}.
\]

\subsection{Proofs of Lemmas~\ref{lem:law} and~\ref{lem:cost}}

\paragraph{Proof of Lemma~\ref{lem:law}.}
In state $k$, all coordinates $S_{kj}$ have positive means $I_{kj}(F)$, so Lemma~\ref{lem:passage} gives $\E_k\tau^+(k,z)<\infty$. In any state $i\ne k$,
\begin{equation}\label{eq:other-duration}
 \sigma(k,z)\le\inf\{n\ge0:S_{ik}(n)\ge\lambda\},\qquad\E_i\sigma(k,z)\le A_F,
\end{equation}
again by Lemma~\ref{lem:passage}. Thus $\sigma(k,z)$ is integrable in every state, and the bound in \eqref{eq:other-duration} is uniform in $z$. The abandonment rule yields this bound even though the evidence for $k$ against $i$ drifts downward in state $i$.

\emph{Certification probabilities.} In state $k$, each likelihood ratio $e^{S_{jk}(n)}$ is a nonnegative martingale with mean one. Stop it at its first crossing of $e^\lambda=4K$ or at a fixed date $m$, whichever comes first. Its expected value is one, so the probability of a crossing by date $m$ is at most $1/(4K)$. Letting $m$ grow and summing over $j\ne k$ gives
\begin{equation}\label{eq:p-diagonal}
 p_k(k,z)\ge1-\frac{K-1}{4K}>\frac34,
\end{equation}
where we used the almost sure termination established above: failing to issue $z$ requires crossing one of the levels in \eqref{eq:minus}. If $i\ne k$, changing measure on the event $\{\tau^+(k,z)<\tau^-(k)\}$ gives
\begin{equation}\label{eq:p-off}
 p_i(k,z)=\E_k\bigl[e^{-S_{ki}(\sigma(k,z))}\one_{\{\tau^+(k,z)<\tau^-(k)\}}\bigr]\le\frac{h_i(z)}{4K\,h_k(z)},
\end{equation}
where the inequality follows from \eqref{eq:plus}. The change of measure follows by summing the identity at each fixed date over all stopping dates, so it requires no optional-stopping argument for an unbounded likelihood ratio.

\emph{The candidate distribution.} Divide row $i$ of \eqref{eq:coefficients} by $h_i(z)$ and write $\Lambda(z)w(z)=\one$, where
\[
 \Lambda_{ik}(z):=\frac{h_k(z)}{h_i(z)}\,p_i(k,z).
\]
By \eqref{eq:p-diagonal} and \eqref{eq:p-off}, each diagonal entry exceeds $3/4$, and the off-diagonal entries of each row sum to less than $1/4$. Hence $\Lambda(z)$ is invertible. Each entry is measurable in $z$, because the stopping events are jointly measurable in the candidate signal and the observations, and matrix inversion is continuous on invertible matrices, so $w(z)=\Lambda(z)^{-1}\one$ is measurable. To bound the solution, write $\Lambda=\Lambda_0(\mathrm{Id}+\Lambda_1)$, with $\Lambda_0$ diagonal and $\|\Lambda_1\|_\infty<1/3$. The Neumann series gives $\|w\|_\infty\le2$. The row equations then imply
\[
 \Lambda_{ii}w_i=1-\sum_{k\ne i}\Lambda_{ik}w_k\ge1-\frac14\|w\|_\infty\ge\frac12>0.
\]
Thus every $w_i$ is positive. The same equations now give $\Lambda_{ii}w_i\le1$, so $w_i\le4/3$ and $w_ip_i(i,z)\le1$. Substituting this sharper upper bound once more yields $\Lambda_{ii}w_i\ge2/3$, and since $\Lambda_{ii}\le1$, $w_i\ge2/3$. Integrating the bounds yields
\[
 \frac{2K}3\le q^{-1}\le\frac{4K}3.
\]
Since $K\ge2$, this proves $0<q<1$ and shows that the candidate distribution has total mass one; it can be sampled with independent randomization. Finally, in state $i$, an attempt issues a signal in $A$ with probability
\[
 \Prob_i(O\in A)=q\sum_k\int_Aw_k(z)\,p_i(k,z)\,H_k(dz)=q\,H_i(A),
\]
where the last equality is \eqref{eq:coefficients}. The attempt terminates almost surely, so $\Prob_i(O=\partial)=1-q$.\hfill$\square$

\paragraph{Proof of Lemma~\ref{lem:cost}.}
The bound \eqref{eq:other-duration} controls the duration when $k\ne i$. It remains to bound the duration in state $i$ when the candidate state is $i$. We show that
\begin{equation}\label{eq:conditional-cost}
 \E_i\sigma(i,z)\le p_i(i,z)\,b_i(z)+A_F\bigl(1+\sqrt{b_i(z)}\bigr).
\end{equation}
The factor $p_i(i,z)$ on the leading term charges the full evidence requirement of the signal only on the paths that issue it.

First consider the paths that abandon. For $j\ne i$, let
\[
 \tau_{ji}:=\inf\{n\ge0:S_{ji}(n)\ge\lambda\}.
\]
In state $j$, $S_{ji}$ has positive drift, so this time has finite mean by Lemma~\ref{lem:passage}. Changing measure separately at each finite date gives
\begin{equation}\label{eq:unsuccessful-time}
 \E_i\bigl[\tau_{ji}\one_{\{\tau_{ji}<\infty\}}\bigr]=\E_j\bigl[\tau_{ji}\,e^{-S_{ji}(\tau_{ji})}\bigr]\le\frac1{4K}\E_j\tau_{ji}\le A_F.
\end{equation}
On the paths that abandon, $\sigma(i,z)$ equals the first such crossing. Summing \eqref{eq:unsuccessful-time} over $j\ne i$ therefore gives
\begin{equation}\label{eq:partial-cost}
 \E_i\bigl[\sigma(i,z)\one_{\{\tau^-(i)\le\tau^+(i,z)\}}\bigr]\le A_F,
\end{equation}
so this part of the expected duration is bounded uniformly, even when the likelihood ratios of $z$ are large.

Next consider a path that issues $z$ at a positive time $\sigma:=\sigma(i,z)$. At least one coordinate in \eqref{eq:plus} was below its level one observation earlier. With $M_{ij}(n):=S_{ij}(n)-nI_{ij}(F)$ and $\Delta S_{ij}(0):=0$, that coordinate gives, as in \eqref{eq:passage-path},
\begin{equation}\label{eq:report-path}
 \sigma\le b_i(z)+A_F+\sum_{j\ne i}\frac{|M_{ij}(\sigma)|+|\Delta S_{ij}(\sigma)|}{I_{ij}(F)},
\end{equation}
where the constant absorbs the margin $\lambda$; at time zero, the inequality is immediate. Lemma~\ref{lem:law} has already established that $\E_i\sigma<\infty$, and \eqref{eq:stopped-square} gives
\[
 \E_iM_{ij}(\sigma)^2=\operatorname{Var}_i\Bigl(\log\frac{dF_i}{dF_j}(X_1)\Bigr)\E_i\sigma,\qquad
 \E_i|\Delta S_{ij}(\sigma)|^2\le\E_i\sum_{t\le\sigma}|\Delta S_{ij}(t)|^2=\E_i|\Delta S_{ij}(1)|^2\,\E_i\sigma.
\]
Multiply \eqref{eq:report-path} by the indicator of issuing $z$, take expectations, and add \eqref{eq:partial-cost}. The Cauchy--Schwarz inequality yields
\[
 \E_i\sigma\le p_i(i,z)\,b_i(z)+A_F+A_F\sqrt{\E_i\sigma}.
\]
Since $\E_i\sigma<\infty$ and $p_i(i,z)\le1$, the implication \eqref{eq:quadratic-step} proves \eqref{eq:conditional-cost}.

It remains to average over the candidate. By \eqref{eq:other-duration}, \eqref{eq:conditional-cost}, and \eqref{eq:a-bounds},
\begin{align*}
 q^{-1}\E_i\nu&=\sum_k\int w_k(z)\,\E_i\sigma(k,z)\,H_k(dz)\\
 &\le\int w_i(z)\,p_i(i,z)\,b_i(z)\,H_i(dz)+A_F\Bigl(1+\int\sqrt{b_i(z)}\,H_i(dz)\Bigr)\\
 &\le\int b_i(z)\,H_i(dz)+A_F\Bigl(1+\int\sqrt{b_i(z)}\,H_i(dz)\Bigr).
\end{align*}
The integrability of $b_i$ implies that of $\sqrt{b_i}$, so all expected durations are finite.\hfill$\square$

\section{Proofs of Proposition~\ref{prop:inclusion}, Theorem~\ref{thm:generic-equivalence}, and Proposition~\ref{prop:reversal}}\label{app:equivalence}

\subsection{Proof of Proposition~\ref{prop:inclusion}}\label{app:inclusion}

The proof uses the following closedness property of the terminal experiments that an experiment generates.

\begin{lemma}[Limits of policies]\label{lem:closed}
Let $F$ satisfy \Cref{ass:regular}, and fix a Borel
identification of the signal space with $[0,1]$.
Suppose $(N_k,Y_k)$ are policies for $F$ such that,
in every state $i$, $\Law_i^F(Y_k)$ converges weakly
to $H_i$ and $\E_i^F N_k\to\bar t_i<\infty$.
Then there is a policy $(N,Y)$ for $F$ with
$H(N,Y)=H$ and $\E_i^F N\le\bar t_i$ for every $i$.
\end{lemma}

\begin{proof}
We work with signals in $[0,1]$. Fix a reference
state $i_0$ and write $P_t:=F_{i_0}^{\otimes t}$,
with $P_0$ the unit mass on the empty history.
For each $k$ and $t\ge0$, let $\gamma_{k,t}$
be the joint subprobability law of $(X^t,Y_k)$
on $\{N_k=t\}$ under state $i_0$, and let
$\alpha_{k,t}$ be its marginal on $X^t$.

The stopping-time property implies
\[
\sum_{s=0}^t\alpha_{k,s}\otimes P_{t-s}
\le P_t \qquad(t\ge0).
\]
Indeed, the $s$th summand is the law of $X^t$
on $\{N_k=s\}$, and these events are disjoint.
Since $[0,1]^{t+1}$ is compact, a diagonal
subsequence gives weak convergence
$\gamma_{k,t}\Rightarrow\gamma_t$ for every $t$.
Writing $\alpha_t$ for the first marginal of
$\gamma_t$, the preceding inequality passes
to the limit. In particular, $\alpha_t\le P_t$.

By disintegration, there are measurable
subprobability kernels $\pi_t(dy\mid x^t)$
such that
$\gamma_t(dx^t,dy)=P_t(dx^t)\pi_t(dy\mid x^t)$.
Set $a_t(x^t):=\pi_t([0,1]\mid x^t)$.
The limiting inequality implies
$\sum_{s\le t}a_s(x^s)\le1$ for $P_t$-almost
every $x^t$. We may choose versions of the
kernels for which this holds along every
history: set $\pi_t$ to zero at and after
the first prefix violating the inequality.
This changes $\gamma_t$ only on null sets.

The kernels define a randomized stopping rule.
Conditional on continuation through date $t-1$,
it stops at date $t$ and reports $dy$ with
probability
$\pi_t(dy\mid x^t)/
(1-\sum_{s<t}a_s(x^s))$,
with arbitrary behavior when the denominator
is zero. The consistency inequality makes
this a subprobability kernel, and induction
shows that its joint law on $\{N=t\}$ under
state $i_0$ is $\gamma_t$.

We next verify the laws under every state.
Fix $i$ and write
$L_{i,t}:=dF_i^{\otimes t}/dP_t$.
Since the stopping rule is state-independent,
the joint laws on $\{N_k=t\}$ and $\{N=t\}$
under state $i$ are respectively
$L_{i,t}\gamma_{k,t}$ and $L_{i,t}\gamma_t$.
Moreover, the former converges weakly to
the latter. To see this, approximate
$L_{i,t}\in L^1(P_t)$ by bounded continuous
functions. Weak convergence applies to
each approximation, and the approximation
errors are uniformly controlled because
the first marginals of $\gamma_{k,t}$ and
$\gamma_t$ are dominated by $P_t$.

Finally, let $B_i:=\sup_k\E_i^F N_k<\infty$.
Markov's inequality gives
$\Prob_i^F(N_k>m)\le B_i/(m+1)$, uniformly
in $k$. Thus convergence at each stopping
date extends to the terminal distribution:
for every continuous $f$ on $[0,1]$,
\[
\int f\,dH_i
=\sum_{t=0}^\infty
\int f(y)L_{i,t}(x^t)\,
\gamma_t(dx^t,dy).
\]
Indeed, the contribution from dates after
$m$ is bounded by
$\|f\|_\infty B_i/(m+1)$ before taking the
limit in $k$. Taking $f=1$ shows that the
constructed rule stops almost surely in
state $i$, and the identity for general
$f$ establishes $\Law_i^F(Y)=H_i$.
Finally, convergence at each stopping date
and monotone convergence give
$\E_i^F N\le\liminf_k\E_i^F N_k=\bar t_i$.
The rule is therefore a policy with the
required properties.
\end{proof}

\begin{proof}[Proof of Proposition~\ref{prop:inclusion}]
Suppose first that $\mathcal P_n^\mu(G)\subseteq\mathcal P_n^\mu(F)$ for every prior $\mu$ and every $n\ge\bar n$. Let $\mathcal D=(\mu,\mathcal A,u)$, let $n\ge\bar n$, and let $(T,a)$ be a policy for $G$ with $\sum_i\mu_i\E_i^GT\le n$. Its terminal experiment belongs to $\mathcal P_n^\mu(G)$, and hence to $\mathcal P_n^\mu(F)$, so some policy $(N,Y)$ for $F$ with $\sum_i\mu_i\E_i^FN\le n$ generates the same terminal experiment. Taking the action $Y$ yields the same expected payoff as $(T,a)$. Hence $V_n(F;\mathcal D)\ge V_n(G;\mathcal D)$.

Conversely, suppose that $V_n(F;\mathcal D)\ge
V_n(G;\mathcal D)$ for every finite decision problem
$\mathcal D$ and every $n\ge\bar n$. Fix such an $n$.

We first consider a prior $\mu$ with full support.
Fix a finite signal set $\mathcal Z$, and let
$\mathcal K$ be the set of terminal experiments
generated by policies for $F$ with signals in
$\mathcal Z$ and expected sample size at most $n$
under $\mu$. Viewed as a subset of
$\R^{\Theta\times\mathcal Z}$, $\mathcal K$ is
convex by randomization and closed by
Lemma~\ref{lem:closed}: full support bounds the
expected sample size in every state, so any
convergent sequence of terminal experiments has
a subsequence of policies to which the lemma
applies.

Suppose some $H\in\mathcal P_n^\mu(G)$ with
signals in $\mathcal Z$ does not belong to
$\mathcal K$. By strict separation, there are
numbers $w_i(z)$ such that
\[
\sum_{i,z}w_i(z)H_i(z)
>
\sup_{H'\in\mathcal K}\sum_{i,z}w_i(z)H'_i(z).
\]
Consider the decision problem with prior $\mu$,
action set $\mathcal Z$, and payoffs
$u(z,i)=w_i(z)/\mu_i$. The left side is attainable
with $G$, whereas the right side is the optimal
value with $F$, contradicting the assumed value
comparison. Thus every finite-valued terminal
experiment in $\mathcal P_n^\mu(G)$ also belongs
to $\mathcal P_n^\mu(F)$.

Now let $H=H(T,Z)\in\mathcal P_n^\mu(G)$ have
an arbitrary standard Borel signal space.
Identify that space with $[0,1]$, and choose
finite-valued measurable quantizations $q_k$
such that $q_k(z)\to z$ uniformly. The policies
$(T,q_k(Z))$ have the same expected sample sizes
as $(T,Z)$, so the preceding argument gives
policies $(N_k,Y_k)$ for $F$ that generate
$H_k:=H(T,q_k(Z))$ within budget $n$.
Since $q_k(Z)\to Z$ uniformly, $H_k$ converges
weakly to $H$ in every state. Full support
bounds $\E_i^F N_k\le n/\mu_i$, so along a
subsequence these expectations converge to
limits $\bar t_i$ satisfying
$\sum_i\mu_i\bar t_i\le n$.
Lemma~\ref{lem:closed} then gives a policy for
$F$ that generates $H$ within budget $n$.
This proves the inclusion for full-support priors.

Finally, let $\mu$ be an arbitrary prior and
fix $H=H(T,Z)\in\mathcal P_n^\mu(G)$.
Choose a full-support prior $\nu$ and set
$M:=\max\{n,\sum_i\nu_i\E_i^GT\}$.
For $\varepsilon\in(0,1)$, let
$\mu^\varepsilon=(1-\varepsilon)\mu+\varepsilon\nu$
and $n^\varepsilon=(1-\varepsilon)n+\varepsilon M$.
The policy $(T,Z)$ is feasible under
$(\mu^\varepsilon,n^\varepsilon)$, with
$n^\varepsilon\ge\bar n$. By the full-support
case, there is a policy for $F$ generating $H$
whose statewise expected sample sizes
$s_i^\varepsilon$ satisfy
\[
(1-\varepsilon)\sum_i\mu_i s_i^\varepsilon
+\varepsilon\sum_i\nu_i s_i^\varepsilon
\le (1-\varepsilon)n+\varepsilon M.
\]
If $\sum_i\nu_i s_i^\varepsilon\ge M$ for
some $\varepsilon$, this policy already meets
the original budget $n$. Otherwise,
$\sum_i\nu_i s_i^\varepsilon<M$ for every
$\varepsilon$, so the $s_i^\varepsilon$ are
uniformly bounded in every state, since $\nu$
has full support. Moreover, the displayed
inequality gives
$\sum_i\mu_i s_i^\varepsilon
\le n+\varepsilon M/(1-\varepsilon)$.
Taking a sequence $\varepsilon\downarrow0$
along which all statewise expected sample sizes
converge, Lemma~\ref{lem:closed} yields a
policy generating $H$ with expected sample size
at most $n$ under $\mu$. Hence
$\mathcal P_n^\mu(G)\subseteq\mathcal P_n^\mu(F)$
for every prior $\mu$ and every $n\ge\bar n$.
\end{proof}

\subsection{Proof of Theorem~\ref{thm:generic-equivalence}}\label{app:thm1}

Section~\ref{subsec:theorem2} shows that (iii) implies (i), using Theorem~\ref{thm:conversion} and Proposition~\ref{prop:inclusion}. We prove that (i) implies (ii) and that (ii) implies the weak inequalities $I_{ij}(F)\ge I_{ij}(G)$; both arguments apply to every pair of experiments. For non-tied pairs, the weak inequalities are strict, which gives (iii).

\begin{lemma}[Bounded budgets leave a loss]\label{lem:loss}
Let $G$ be an experiment whose distributions are mutually absolutely continuous, and let $\bar n\ge0$. There is $\kappa>0$, which depends only on $G$ and $\bar n$, such that
\[
 V^*(\mathcal D)-V_{\bar n}(G;\mathcal D)\ge\kappa\,v(\mathcal D)\qquad\text{for every finite decision problem }\mathcal D.
\]
\end{lemma}

\begin{proof}
Let $m:=\lfloor2\bar n\rfloor+1$, and for $\gamma>0$, let $A_\gamma$ be the set of signals $x$ of $G$ with $|\log(dG_i/dG_j)(x)|\le\gamma$ for all states $i$ and $j$. The log-likelihood ratios are finite almost everywhere, so we can choose $\gamma$ with $m\max_iG_i(A_\gamma^c)\le1/4$. Let $\mathcal D=(\mu,\mathcal A,u)$, let $(T,a)$ be a policy for $G$ with $\sum_i\mu_i\E_i^GT\le\bar n$, and write $\Prob^\mu:=\sum_i\mu_i\Prob_i^G$ and $\E^\mu$ for the corresponding expectation. Let $B$ be the event that $T\le m$ and $X_1,\ldots,X_T\in A_\gamma$. Then
\[
 \Prob^\mu(B^c)\le\Prob^\mu(T>m)+\sum_{t=1}^m\Prob^\mu(X_t\notin A_\gamma)\le\frac{\bar n}m+m\max_iG_i(A_\gamma^c)\le\frac34,
\]
so $\Prob^\mu(B)\ge1/4$. On $B$, $S^G_{ji}(T)\le\gamma m$ for all states $i$ and $j$, so by Lemma~\ref{lem:likelihood}, $\E_j^G[W\one_B]=\E_i^G[W\one_Be^{S^G_{ji}(T)}]\le e^{\gamma m}\,\E_i^G[W\one_B]$ for every nonnegative random variable $W$ that is determined by the observations up to time $T$ and the randomization. Applying this inequality with $W:=\ell(a,i)$, multiplying by $\mu_i\mu_j$, and summing over $i$ and $j$ gives
\[
 \sum_i\mu_i\E_i^G\bigl[\ell(a,i)\bigr]\ge\sum_i\mu_i\E_i^G\bigl[\ell(a,i)\one_B\bigr]\ge e^{-\gamma m}\,\E^\mu\Bigl[\one_B\sum_i\mu_i\ell(a,i)\Bigr]\ge e^{-\gamma m}\,\Prob^\mu(B)\,v(\mathcal D),
\]
where the last inequality uses $\sum_i\mu_i\ell(b,i)\ge v(\mathcal D)$ for every action $b$. The left side is $V^*(\mathcal D)$ minus the expected payoff of $(T,a)$. Taking the supremum over such policies proves the lemma with $\kappa:=e^{-\gamma m}/4$.
\end{proof}

\begin{lemma}[From budgets to costs]\label{lem:budget-cost}
Let $G$ be an experiment whose distributions are distinct and mutually absolutely continuous, and let $\bar n\ge0$. For all $\underline v>0$ and $\bar\ell<\infty$, there is $\bar c>0$, which depends only on $G$, $\bar n$, $\underline v$, and $\bar\ell$, with the following property. If $F$ is an experiment and $\mathcal D$ is a finite decision problem with $v(\mathcal D)\ge\underline v$, $\ell^*(\mathcal D)\le\bar\ell$, and $V_n(F;\mathcal D)\ge V_n(G;\mathcal D)$ for every $n\ge\bar n$, then $W_c(F;\mathcal D)\ge W_c(G;\mathcal D)$ for every $c\in(0,\bar c)$.
\end{lemma}

\begin{proof}
Let $\kappa$ be the constant of Lemma~\ref{lem:loss}, choose an integer $m\ge1$ with $(K-1)\,\bar\ell\,\eta_G^m\le\kappa\underline v/2$, where $\eta_G<1$ is defined in Lemma~\ref{lem:ml}, and let $\bar c:=\kappa\underline v/(2m)$. Fix $c\in(0,\bar c)$. Because $V_n$ is nondecreasing in $n$, Lemma~\ref{lem:loss} gives, for every $n<\bar n$,
\[
 V_n(G;\mathcal D)-cn\le V_{\bar n}(G;\mathcal D)\le V^*(\mathcal D)-\kappa\,v(\mathcal D)\le V^*(\mathcal D)-\kappa\underline v.
\]
By Lemma~\ref{lem:ml}, $W_c(G;\mathcal D)\ge V^*(\mathcal D)-\kappa\underline v/2-cm>V^*(\mathcal D)-\kappa\underline v$. Hence budgets below $\bar n$ do not approach the supremum in \eqref{eq:duality} for $G$, and
\[
 W_c(G;\mathcal D)=\sup_{n\ge\bar n}\bigl\{V_n(G;\mathcal D)-cn\bigr\}\le\sup_{n\ge\bar n}\bigl\{V_n(F;\mathcal D)-cn\bigr\}\le W_c(F;\mathcal D).\qedhere
\]
\end{proof}

\begin{proof}[Proof that (i) implies (ii)]
Suppose that $F$ large-budget dominates $G$ with threshold $\bar n$, and let $\mathcal D$ be a finite decision problem. If $v(\mathcal D)=0$, some action is optimal in every state with positive prior probability, so stopping at once and taking it attains $V^*(\mathcal D)$ with either experiment, and $W_c(F;\mathcal D)=W_c(G;\mathcal D)=V^*(\mathcal D)$ for every $c$. If $v(\mathcal D)>0$, Lemma~\ref{lem:budget-cost} with $\underline v:=v(\mathcal D)$ and $\bar\ell:=\ell^*(\mathcal D)$ gives $\bar c_{\mathcal D}>0$ with $W_c(F;\mathcal D)\ge W_c(G;\mathcal D)$ for every $c\in(0,\bar c_{\mathcal D})$.
\end{proof}

\begin{lemma}[Testing two states at a small cost]\label{lem:pair-testing}
Let $E$ be an experiment, let $i$ and $j$ be distinct states, and let $\mathcal D$ be a finite decision problem whose prior assigns positive probabilities $\mu_i$ and $\mu_j$ to $i$ and $j$ and zero probability to the other states, and in which each of $i$ and $j$ has a unique optimal action, the two being distinct. Then
\[
 \lim_{c\downarrow0}\frac{V^*(\mathcal D)-W_c(E;\mathcal D)}{c\log(1/c)}=\frac{\mu_i}{I_{ij}(E)}+\frac{\mu_j}{I_{ji}(E)}.
\]
\end{lemma}

\begin{proof}
Write $R_c:=V^*(\mathcal D)-W_c(E;\mathcal D)$ and $\varphi:=\log(1/c)$, and let $a_i\ne a_j$ be the optimal actions in states $i$ and $j$.

For the upper bound, stop when $S^E_{ij}$ first leaves $(-\varphi,\varphi)$, and take $a_i$ if it exits above and $a_j$ if it exits below. In state $i$, the stopping time is at most the first time at which $S^E_{ij}$ reaches $\varphi$, and in state $j$, at most the first time at which $S^E_{ji}$ reaches $\varphi$, so Lemma~\ref{lem:passage} gives
\[
 \E_iT\le\frac{\varphi}{I_{ij}(E)}+O\bigl(\sqrt\varphi+1\bigr),\qquad\E_jT\le\frac{\varphi}{I_{ji}(E)}+O\bigl(\sqrt\varphi+1\bigr).
\]
The stopping time also has finite mean in every other state. Because $I_{ij}(E)>0$, there is $b>0$ such that an observation raises $S^E_{ij}$ by at least $b$ with positive probability in state $i$, and hence, by mutual absolute continuity, in every state. A run of more than $2\varphi/b$ such observations forces exit from any point of $(-\varphi,\varphi)$, so, by considering independent blocks of observations, the stopping time is bounded in every state by a constant multiple of a geometric random variable. In state $i$, the action is wrong only if $e^{S^E_{ji}}$, a nonnegative martingale with mean one in state $i$, reaches $e^\varphi$; this has probability at most $e^{-\varphi}=c$, and similarly in state $j$. The losses from wrong actions therefore contribute $O(c)$, and $R_c\le c\varphi(\mu_i/I_{ij}(E)+\mu_j/I_{ji}(E))+O(c\sqrt\varphi+c)$.

For the lower bound, take a policy $(T,a)$ whose expected loss, including the cost of observations, is at most $R_c+c^2$. By the upper bound, this expected loss is $O(c\varphi)$. Because $\mu_i,\mu_j>0$ and every action other than $a_i$ has a positive loss in state $i$, and similarly for $j$, both $\alpha:=\Prob_i(a\ne a_i)$ and $\beta:=\Prob_j(a=a_i)$ are $O(c\varphi)$. Lemma~\ref{lem:likelihood}, applied to the indicator of $\{a=a_i\}$, and \eqref{eq:bernoulli-bound} give
\[
 I_{ij}(E)\,\E_iT\ge\KL{\operatorname{Ber}(1-\alpha)}{\operatorname{Ber}(\beta)}\ge(1-\alpha)\log(1/\beta)-\log2\ge\varphi-o(\varphi),
\]
since $\log(1/\beta)\ge\varphi-\log\varphi-O(1)$ and $\alpha\to0$; here $0<\beta<1$ for small $c$, because $1-\alpha>0$ and the distributions of a terminal experiment are mutually absolutely continuous. The same argument, applied to the indicator of $\{a=a_j\}$, whose probabilities $\Prob_j(a\ne a_j)$ and $\Prob_i(a=a_j)\le\alpha$ are also $O(c\varphi)$, gives $I_{ji}(E)\,\E_jT\ge\varphi-o(\varphi)$. Hence $R_c+c^2\ge c(\mu_i\E_iT+\mu_j\E_jT)\ge c\varphi(\mu_i/I_{ij}(E)+\mu_j/I_{ji}(E))(1-o(1))$, which proves the lemma.
\end{proof}

\begin{proof}[Proof that (ii) implies (iii)]
Fix distinct states $i$ and $j$ and $p\in(0,1)$, and let $\mathcal D_p$ be the decision problem with prior $\mu_i=p$ and $\mu_j=1-p$, two actions $a_i$ and $a_j$, and payoff one for $a_i$ in state $i$ and for $a_j$ in state $j$, and zero otherwise. If $F$ small-cost dominates $G$, then $V^*(\mathcal D_p)-W_c(F;\mathcal D_p)\le V^*(\mathcal D_p)-W_c(G;\mathcal D_p)$ for all small $c$, and Lemma~\ref{lem:pair-testing} gives
\[
 \frac p{I_{ij}(F)}+\frac{1-p}{I_{ji}(F)}\le\frac p{I_{ij}(G)}+\frac{1-p}{I_{ji}(G)}.
\]
Letting $p\to1$ yields $I_{ij}(F)\ge I_{ij}(G)$. Thus small-cost dominance implies the weak inequalities for every pair of experiments, and the strict inequalities of KL dominance for non-tied pairs.
\end{proof}

\subsection{Proof of Proposition~\ref{prop:reversal}}\label{app:reversal}

Write $V_n^{\mathrm{fix}}(E;\mathcal D)$ for the optimal expected payoff in $\mathcal D$ from $n$ independent observations of $E$ followed by an action.

\begin{proof}[Proof of Proposition~\ref{prop:reversal}]
Let
$\eta_G:=\max_{i\ne j}\int\sqrt{dG_i\,dG_j}<1$,
and set $\bar\alpha_G:=1-\eta_G$.
Fix $\alpha<\bar\alpha_G$ and any experiment $E$.

Consider a decision problem
$\mathcal D=(\mu,\mathcal A,u)$ with positive
value of information. Write
$v:=V^*(\mathcal D)-\max_a\sum_i\mu_i u(a,i)>0$,
and let $L$ be the largest loss from any action
in any state with positive prior probability.

With $n$ observations of $E^\alpha$, the
probability that every observation is uninformative
is $(1-\alpha)^n$, independently of the state.
Conditional on this event, the decision maker
has learned nothing and therefore incurs an
expected loss of at least $v$ relative to full
information. On the other hand, by
Lemma~\ref{lem:ml}, taking $n$ observations of
$G$ and choosing the action optimal for the
maximum-likelihood state yields an expected
loss of at most $(K-1)L\eta_G^n$.
Since $\eta_G<1-\alpha$, for all sufficiently
large $n$,
\[
V^*(\mathcal D)-V_n^{\mathrm{fix}}(G;\mathcal D)
\le (K-1)L\eta_G^n
< v(1-\alpha)^n
\le V^*(\mathcal D)-V_n^{\mathrm{fix}}(E^\alpha;\mathcal D).
\]
Thus $G$ strictly dominates $E^\alpha$ in the
MS order, regardless of $E$.

Conversely, fix any $\alpha\in(0,1)$.
Choose an integer $m>1/\alpha$ and let
$E=G^{\otimes m}$. By
Lemma~\ref{lem:mixing},
$I_{ij}(E^\alpha)=\alpha m I_{ij}(G)>I_{ij}(G)$
for every $i\ne j$. Since both experiments
satisfy \Cref{ass:regular},
Theorem~\ref{thm:generic-equivalence} implies
that $E^\alpha$ large-budget dominates and
small-cost dominates $G$.

Finally, when $\alpha<\bar\alpha_G$, the
constructed experiments are Blackwell-incomparable
at every fixed sample size. Indeed, the Hellinger
affinity between any two states under $E^\alpha$
is at least $1-\alpha>\eta_G$, exceeding that
under $G$. Since Hellinger affinity is
multiplicative under independent sampling and
cannot decrease under garbling, $G^{\otimes n}$
cannot be a garbling of $(E^\alpha)^{\otimes n}$.
The reverse is ruled out by the strict KL
inequalities and the data-processing inequality.
\end{proof}

\section{Proofs for Section~\ref{sec:costs}}\label{app:costs}

For a source experiment $E$, write $\pi^E_{ij}:=c\mu_i/I_{ij}(E)$, and for a nonnegative array $x=(x_{ij})_{i\ne j}$, let
\[
 f_\pi(x):=\sum_{i\in\Theta}\max_{j\ne i}\pi_{ij}x_{ij}.
\]
The function $f_\pi$ is continuous, nondecreasing in $x$ and in $\pi$, positively homogeneous, subadditive, and convex. For a target experiment $H$, let $I(H):=(I_{ij}(H))_{i\ne j}$, so that \eqref{eq:per-copy-formula} reads $C_\infty^E(H)=f_{\pi^E}(I(H))$.

\subsection{Proof of Theorem~\ref{thm:costs}(a)}

For any policy $(T,Z)$ that generates $H^{\otimes m}$,
data processing and Wald's identity give
$mI_{ij}(H)\le I_{ij}(E)\E_i^E T$ for every $i\ne j$.
Hence $\E_i^E T\ge m\rho_i(E\to H)$ in every state $i$.
Conversely, applying Theorem~\ref{thm:conversion} to
the policy that takes $m$ observations of $H$ and
reports the entire sample yields a policy for $E$
that generates $H^{\otimes m}$ with expected sample
size at most $m\rho_i(E\to H)+A_{E,H}\sqrt m$
in every state $i$. Thus,
\[
c\sum_i\mu_i\rho_i(E\to H)
\le \frac{C^E(H^{\otimes m})}{m}
\le c\sum_i\mu_i\rho_i(E\to H)
+\frac{cA_{E,H}}{\sqrt m}.
\]
Since $A_{E,H}$ is independent of $m$, both the
limit and the infimum in \eqref{eq:cost-function}
equal the expression in \eqref{eq:per-copy-formula}.
\hfill$\square$

\subsection{Proof of Theorem~\ref{thm:costs}(b)}

The converse requires source experiments whose KL divergences form an arbitrary array of positive numbers.

\begin{lemma}[Prescribed KL divergences]\label{lem:rates}
For every array $(x_{ij})_{i\ne j}$ of positive numbers, there is an experiment $E$ with finitely many signals, each with positive probability in every state, such that $I_{ij}(E)=x_{ij}$ for all distinct $i$ and $j$. In particular, $E$ satisfies Assumption~\ref{ass:regular}.
\end{lemma}

\begin{proof}
Collect KL divergences into vectors $I(\cdot)\in\R^{K(K-1)}$ indexed by ordered pairs, and let $e_{ij}$ denote the unit vector for the pair $(i,j)$. The construction has one parameter $M\ge8$, fixed at the end.

\emph{The components.} For each ordered pair $(i,j)$, let $E^{ij}$ be the experiment with two signals, success and failure, whose success probabilities are
\[
 q_i:=M^{-1},\qquad q_j:=M^{-1}e^{-M^2},\qquad q_k:=M^{-2}\quad(k\notin\{i,j\}).
\]
All lie strictly between zero and one, so each signal has positive probability in every state; for $K\ge3$, the states outside $\{i,j\}$ share a distribution, which Lemma~\ref{lem:mixing} allows. Each KL divergence splits into a success term and a failure term,
\[
 I_{k\ell}(E^{ij})=q_k\log\frac{q_k}{q_\ell}+(1-q_k)\log\frac{1-q_k}{1-q_\ell},
\]
and $|\log(1-q)|\le2q$ for $q\le1/2$ bounds every failure term by $2\max_kq_k\le2/M$ in absolute value. For the pair $(i,j)$, the success term is $M^{-1}\log e^{M^2}=M$, so
\begin{equation}\label{eq:pair-rates}
 I_{ij}(E^{ij})\ge M-\frac2M\ge\frac M2.
\end{equation}
For every other ordered pair, the success term is at most one: if $\ell\ne j$, then $q_\ell\ge M^{-2}$ and $q_k\le M^{-1}$ give $q_k\log(q_k/q_\ell)\le M^{-1}\log M\le\tfrac12$; and if $\ell=j$ and $k\notin\{i,j\}$, it equals $M^{-2}(M^2-\log M)=1-M^{-2}\log M$. In the first case, the failure term adds at most $2/M\le\tfrac14$; in the second, it is nonpositive, because $q_k>q_j$. Either way, since $M\ge8$,
\begin{equation}\label{eq:pair-rates-other}
 I_{k\ell}(E^{ij})\le1\qquad\text{for }(k,\ell)\ne(i,j).
\end{equation}
Relabelling the states carries $E^{ij}$ into $E^{\sigma(i)\sigma(j)}$ and permutes the entries accordingly, so \eqref{eq:pair-rates} and \eqref{eq:pair-rates-other} hold uniformly over ordered pairs. Dividing them, every entry of $I(E^{ij})/I_{ij}(E^{ij})-e_{ij}$ lies in $[0,2/M]$, uniformly in $(i,j)$.

\emph{The mixture.} Let $\Phi_M$ be the square matrix of order $K(K-1)$ whose column indexed by $(i,j)$ is $I(E^{ij})/I_{ij}(E^{ij})$. By the last observation, $\Phi_M\to\mathrm{Id}$ as $M\to\infty$, and inversion is continuous at the identity, so $\Phi_M$ is invertible for all large $M$ and $v:=\Phi_M^{-1}x\to x$. Since $x$ has positive entries, so does $v$ for all large $M$; and since $v$ is bounded and $I_{ij}(E^{ij})\ge M/2$ by \eqref{eq:pair-rates}, $\sum_{i\ne j}v_{ij}/I_{ij}(E^{ij})\to0$. Fix an $M$ at which both hold. The weights $w_{ij}:=v_{ij}/I_{ij}(E^{ij})$ are then positive with sum below one. Let $E$ be the labelled mixture of the experiments $E^{ij}$, with weights $w_{ij}$, and of an uninformative experiment, with the remaining weight. By Lemma~\ref{lem:mixing},
\[
 I(E)=\sum_{i\ne j}w_{ij}\,I(E^{ij})=\sum_{i\ne j}v_{ij}\frac{I(E^{ij})}{I_{ij}(E^{ij})}=\Phi_Mv=x.
\]
Every signal of $E$ has positive probability in every state, and its distributions are distinct because every KL divergence is positive.
\end{proof}

\begin{proof}[Proof of Theorem~\ref{thm:costs}(b)]
If $\Phi=C_\infty^E$, part (a) gives
$\Phi(H)=f_{\pi^E}(I(H))$, with positive weights
$\pi^E_{ij}=c\mu_i/I_{ij}(E)$.
Conversely, given positive weights $\pi$, Lemma~\ref{lem:rates}
yields an experiment $E$ with
$I_{ij}(E)=c\mu_i/\pi_{ij}$, so
$C_\infty^E(H)=f_\pi(I(H))$ for every target $H$.

For the ordering claim, suppose
$f_\pi(I(H))\le f_{\widetilde\pi}(I(H))$
for every target $H$.
By Lemma~\ref{lem:rates}, every strictly positive array
is the KL array of some target experiment.
Continuity therefore extends the inequality
$f_\pi(x)\le f_{\widetilde\pi}(x)$
to every nonnegative array $x$.
In particular, taking $x=e_{ij}$ gives
$\pi_{ij}=f_\pi(e_{ij})
\le f_{\widetilde\pi}(e_{ij})=\widetilde\pi_{ij}$.
Uniqueness follows by applying this implication
in both directions.
\end{proof}

\subsection{Proof of Corollary~\ref{cor:cost-ranking}}

\begin{proof}[Proof of Corollary~\ref{cor:cost-ranking}]
If $C_\infty^F(H)\le C_\infty^G(H)$ for every
target $H$, Theorem~\ref{thm:costs}(b) implies
$\pi^F_{ij}\le\pi^G_{ij}$, or equivalently,
$I_{ij}(F)\ge I_{ij}(G)$, for every $i\ne j$.
Since $F$ and $G$ are non-tied, all these
inequalities are strict, and
Theorem~\ref{thm:generic-equivalence} gives
large-budget dominance.

Conversely, large-budget dominance implies
$I_{ij}(F)\ge I_{ij}(G)$ for every $i\ne j$,
as shown in Appendix~\ref{app:thm1}.
Thus $\pi^F\le\pi^G$ entrywise, and
Theorem~\ref{thm:costs}(a) gives
$C_\infty^F(H)\le C_\infty^G(H)$ for every
target $H$.
\end{proof}

\subsection{Economies of scope}\label{app:scope}

Fix a source experiment $E$, target experiments $H$ and $H'$, and a state $i$, and write $\alpha_j:=I_{ij}(H)/I_{ij}(E)$ and $\beta_j:=I_{ij}(H')/I_{ij}(E)$ for $j\ne i$. Because $I_{ij}(H\otimes H')=I_{ij}(H)+I_{ij}(H')$ by Lemma~\ref{lem:mixing}, $\rho_i(E\to H\otimes H')=\max_j(\alpha_j+\beta_j)$, and the saving
\[
 \Delta_i:=\rho_i(E\to H)+\rho_i(E\to H')-\rho_i(E\to H\otimes H')=\max_j\alpha_j+\max_j\beta_j-\max_j(\alpha_j+\beta_j)
\]
is nonnegative. It is zero if and only if $H$ and $H'$ have a common bottleneck in state $i$. Indeed, if $j^*$ attains $\max_j(\alpha_j+\beta_j)$ and $\Delta_i=0$, then $\alpha_{j^*}+\beta_{j^*}=\max_j\alpha_j+\max_j\beta_j$, which forces $\alpha_{j^*}=\max_j\alpha_j$ and $\beta_{j^*}=\max_j\beta_j$; conversely, a common bottleneck $j$ gives $\max_l(\alpha_l+\beta_l)\ge\alpha_j+\beta_j=\max_l\alpha_l+\max_l\beta_l$. By \eqref{eq:per-copy-formula}, $C_\infty^E(H)+C_\infty^E(H')-C_\infty^E(H\otimes H')=c\sum_i\mu_i\Delta_i$, which is positive if and only if $\Delta_i>0$ for some state $i$ with $\mu_i>0$, that is, if and only if the two targets have no common bottleneck in some state.

With three or more states, every cost function of Theorem~\ref{thm:costs} has economies of scope for some pair of targets. Fix a state $i$ with $\mu_i>0$, two distinct alternatives $j$ and $k$, and $p\in(1/2,1)$. Let $H$ be the target experiment with two signals whose distribution is $\operatorname{Ber}(p)$ in state $j$ and $\operatorname{Ber}(1/2)$ in every other state, and let $H'$ be defined in the same way with $k$ in place of $j$. In state $i$, the only positive numbers among the $\alpha_l$ and $\beta_l$ are $\alpha_j$ and $\beta_k$, so $\max_l(\alpha_l+\beta_l)=\max\{\alpha_j,\beta_k\}$ and $\Delta_i=\min\{\alpha_j,\beta_k\}>0$. Since $\mu_i>0$, the saving $c\sum_l\mu_l\Delta_l$ is positive. With two states, each state has a single alternative, bottlenecks always coincide, and the cost is additive.

\clearpage
\part*{Online Appendix}
\addcontentsline{toc}{part}{Online Appendix}
\setcounter{section}{0}
\renewcommand{\thesection}{OA.\arabic{section}}
\renewcommand{\theHsection}{online.\arabic{section}}
\numberwithin{proposition}{section}

\section{The square-root remainder}\label{app:tightness}

\subsection{A counterexample with multiple bottlenecks}\label{app:multiple-bottlenecks}

We first construct the experiments of Remark~\ref{rem:tight}. They show that the weak inequalities $I_{ij}(F)\ge I_{ij}(G)$ do not imply large-budget dominance, and that the remainder in Theorem~\ref{thm:conversion} cannot in general be reduced to $o(\sqrt{\E_i^GT})$. The two experiments have the same KL divergences for every pair that involves state $0$; they differ in how often evidence against the two other states arrives together. We first establish a passage-time expansion.

\begin{lemma}[Simultaneous-boundary expansion]\label{lem:simultaneous-boundaries}
Let $(Y_t)$ be bounded, independent, and identically distributed random vectors with component means $d_\ell>0$ and covariance matrix $\Sigma$. For $m>0$, let
\[
 \tau_m:=\inf\Bigl\{n\ge0:\sum_{t=1}^nY_{t\ell}\ge md_\ell\text{ for every }\ell\Bigr\}.
\]
If $\xi$ is a centered normal random vector with covariance matrix $\Sigma$, then
\[
 \E\tau_m=m+\chi\sqrt m+o(\sqrt m),\qquad\chi:=\E\max_\ell\{-\xi_\ell/d_\ell\},\qquad\E\tau_m^2=O(m^2).
\]
With two components, $\chi=\sqrt{\operatorname{Var}(\xi_1/d_1-\xi_2/d_2)/(2\pi)}$.
\end{lemma}

\begin{proof}
Write $R_n:=\sum_{t=1}^n(Y_t-d)$ and $g(y):=\max_\ell(-y_\ell/d_\ell)$. At $\tau_m$, every component meets its boundary, whereas at $\tau_m-1$ some component does not. Since the increments are bounded, a constant $M_0$ that does not depend on $m$ satisfies
\begin{equation}\label{eq:boundary-stopping-relation}
 0\le\tau_m-m-g(R_{\tau_m})\le M_0.
\end{equation}
For $n\ge m$, the event $\{\tau_m>n\}$ implies that some component is below its boundary at $n$. Hoeffding's inequality and a union bound therefore give $\Prob(\tau_m>n)\le C_1\exp\{-c_0(n-m)^2/n\}$ for constants $C_1$ and $c_0>0$. For each $\epsilon\in(0,1)$, crossing a given boundary at some $n\le(1-\epsilon)m$ requires an upward deviation of at least $\epsilon md_\ell$, and Hoeffding's inequality and a union bound over these times make this probability exponentially small in $m$. Consequently, $\tau_m/m\to1$ in probability. Lemma~\ref{lem:passage} with levels $md_\ell$ gives $\E\tau_m\le m+C\sqrt m+C$, and the bound on the upper tail gives $\E\tau_m^2=O(m^2)$.

The multivariate central limit theorem gives $R_{\lfloor m\rfloor}/\sqrt m\Rightarrow\xi$. A maximal inequality for the centered increments in a window of length $\epsilon m$ around $m$, together with $\tau_m/m\to1$, gives $R_{\tau_m}/\sqrt m\Rightarrow\xi$. Finally, the martingale identity for squares at $\tau_m\wedge n$ and Fatou's lemma imply $\E\|R_{\tau_m}\|^2\le\operatorname{tr}(\Sigma)\,\E\tau_m=O(m)$. Thus $g(R_{\tau_m})/\sqrt m$ is uniformly integrable, and taking expectations in \eqref{eq:boundary-stopping-relation} proves the expansion. For two components, the identity $\max\{y_1,y_2\}=(y_1+y_2+|y_1-y_2|)/2$ and the absolute moment of a normal random variable give the expression for $\chi$.
\end{proof}

Fix $a>0$. For $\omega\in[0,1)$, define an experiment with three states and signals in $\{-1,1\}^2$ by
\begin{equation}\label{eq:four-signal-family}
 E_0^\omega(x_1,x_2):=\frac{1+\omega x_1x_2}4,\qquad E_j^\omega(x_1,x_2):=E_0^\omega(x_1,x_2)\frac{e^{ax_j}}{\cosh a},\quad j=1,2.
\end{equation}
Each coordinate is uniform in state $0$, so these are probability distributions. They have full support and differ across states. Fix $\omega\in(0,1)$, set $F:=E^0$ and $G:=E^\omega$, and write
\[
 b_0:=\log\cosh a,\qquad c_1:=a\tanh a-b_0.
\]
We show that $I_{ij}(F)\ge I_{ij}(G)$ for every ordered pair and that $\rho_i(F\to G)=1$ in every state, and we construct policies $(T_m,A_m)$ for $G$ with two signals such that
\[
 \E_0^GT_m=m+\frac{a\sqrt{1-\omega}}{b_0\sqrt\pi}\sqrt m+o(\sqrt m),\qquad\E_0^FN_m\ge m+\frac a{b_0\sqrt\pi}\sqrt m+o(\sqrt m)
\]
for every policy $(N_m,Y_m)$ for $F$ that generates the same terminal experiment.

\emph{KL divergences.} Under $E^\omega$, the evidence for state $0$ against state $j$ in one observation is $b_0-aX_j$. Both coordinates have mean zero in state $0$, while in state $j$, $\E_jX_j=\tanh a$ and $\E_jX_{3-j}=\omega\tanh a$. Since $b_0+c_1=a\tanh a$, the arrays of KL divergences, with rows indexed by $i$ and columns by $j$, are
\[
 I(F)=\begin{pmatrix}0&b_0&b_0\\ c_1&0&b_0+c_1\\ c_1&b_0+c_1&0\end{pmatrix},\qquad
 I(G)=\begin{pmatrix}0&b_0&b_0\\ c_1&0&(1-\omega)(b_0+c_1)\\ c_1&(1-\omega)(b_0+c_1)&0\end{pmatrix}.
\]
Moreover, $\operatorname{Var}_0\bigl((b_0-aX_1)-(b_0-aX_2)\bigr)=2a^2(1-\omega)$, so Lemma~\ref{lem:simultaneous-boundaries} gives the constants
\[
 \chi_F=\frac a{b_0\sqrt\pi},\qquad\chi_G=\frac{a\sqrt{1-\omega}}{b_0\sqrt\pi}.
\]

\emph{The policies.} Under $G$, stop and report $1$ when both $S^G_{01}$ and $S^G_{02}$ reach $mb_0$, and stop and report $0$ when either $S^G_{10}$ or $S^G_{20}$ reaches $mc_1/2$. These conditions cannot hold simultaneously. The resulting policy $(T_m,A_m)$ has finite expected duration in every state, since its duration is bounded by a passage time with positive drift. A change of measure at the stopping time gives
\begin{equation}\label{eq:four-signal-errors}
 \Prob_0^G(A_m=0)\le2e^{-mc_1/2},\qquad\Prob_j^G(A_m=1)\le e^{-mb_0},\quad j=1,2.
\end{equation}
Let $\tau_m^G$ be the first time at which both $S^G_{01}$ and $S^G_{02}$ reach $mb_0$. Then $T_m\le\tau_m^G$, with equality when $A_m=1$. Lemma~\ref{lem:simultaneous-boundaries}, \eqref{eq:four-signal-errors}, and the Cauchy--Schwarz inequality give
\[
 \E_0^GT_m=m+\chi_G\sqrt m+o(\sqrt m).
\]

\emph{Every reproduction takes longer.} Let $(N_m,Y_m)$ be a policy for $F$ that generates the terminal experiment of $(T_m,A_m)$, and let $\Gamma:=\{Y_m=1\}$ and $\delta_m:=m^{1/4}$. For $j=1,2$, the change of measure and \eqref{eq:four-signal-errors} imply
\[
 \Prob_0^F\bigl(\Gamma\cap\{S^F_{0j}(N_m)<mb_0-\delta_m\}\bigr)\le e^{mb_0-\delta_m}\,\Prob_j^F(\Gamma)\le e^{-\delta_m}.
\]
Thus, except on an event of probability at most $2e^{-mc_1/2}+2e^{-\delta_m}$, $N_m$ is at least the first time $\hat\tau_m^F$ at which both $S^F_{01}$ and $S^F_{02}$ reach $mb_0-\delta_m$. Lemma~\ref{lem:simultaneous-boundaries}, applied at $m-\delta_m/b_0$, and $\delta_m=o(\sqrt m)$ give
\[
 \E_0^F\hat\tau_m^F=m+\chi_F\sqrt m+o(\sqrt m),\qquad\E_0^F(\hat\tau_m^F)^2=O(m^2).
\]
The Cauchy--Schwarz inequality therefore yields $\E_0^FN_m\ge m+\chi_F\sqrt m+o(\sqrt m)$, with a remainder that is uniform over reproducing policies, because the preceding bounds do not depend on the policy. Since $\chi_F>\chi_G$, every such policy uses at least $\E_0^GT_m+\kappa\sqrt{\E_0^GT_m}$ observations in expectation in state $0$ for large $m$, for any $\kappa<\chi_F-\chi_G$.

\emph{Large-budget dominance fails.} Let $n_m:=\E_0^GT_m$, and let $\mu$ be the prior that assigns probability one to state $0$. The terminal experiment of $(T_m,A_m)$ belongs to $\mathcal P_{n_m}^\mu(G)$ but, for large $m$, not to $\mathcal P_{n_m}^\mu(F)$. Since $n_m\to\infty$, $F$ does not production dominate $G$, and by Proposition~\ref{prop:inclusion}, $F$ does not large-budget dominate $G$. In state $0$, the two experiments supply evidence against states $1$ and $2$ at the same average rates; the positive correlation between the two kinds of evidence under $G$ shortens the wait until both requirements are met.

\subsection{Unique bottlenecks}\label{app:unique-bottleneck}

The preceding example has two binding evidence requirements in state $0$. When the bottleneck is unique, the square-root remainder in Theorem~\ref{thm:conversion} can instead be reduced to a constant. A related distinction appears in multihypothesis testing: \citet{DragalinTartakovskyVeeravalli2000} obtain constant-order corrections when the KL-nearest alternative is unique, whereas competing alternatives can produce square-root corrections. Here uniqueness concerns a ratio of KL divergences across two experiments, and the bound applies uniformly to the exact reproduction of arbitrary stopping policies, rather than to a particular sequential test.

\begin{proposition}[Unique-bottleneck refinement]\label{prop:unique-bottleneck}
Let $F$ and $G$ satisfy Assumption~\ref{ass:regular}. Suppose that, in each state $i$, the maximizer $j_i^*$ in the definition of $\rho_i(F\to G)$ is unique. There is a constant $C_{F,G}<\infty$ such that every policy $(T,Z)$ for $G$ can be reproduced exactly by a policy $(N,Y)$ for $F$ satisfying
\begin{equation}\label{eq:unique-conversion}
\E_i^F N\le\rho_i(F\to G)\,\E_i^G T+C_{F,G}\qquad\text{for every }i.
\end{equation}
\end{proposition}

\begin{proof}
We first assume $T\ge1$ and reproduce the stopped record $R$ by the construction in Appendix~\ref{app:conversion-record}. Fix a state $i$, write $j^*=j_i^*$, $d_j:=I_{ij}(F)$, $g_j:=I_{ij}(G)$, and $\rho:=g_{j^*}/d_{j^*}$. For each record $R$, set $c_j:=S^G_{ij}(T)+\lambda$ and let $b:=c_{j^*}^+$. On an independent sequence of $F$-observations, put $\tau_*:=\inf\{n\ge0:S^F_{ij^*}(n)\ge b\}$. Lorden's uniform overshoot bound \citep{Lorden1970} and Wald's identity give $\E_i^F[S^F_{ij^*}(\tau_*)-b]\le C_F$ and $\E_i^F\tau_*\le b/d_{j^*}+C_F$, conditional on $R$.

We use two standard bounds for centered random walks $M$ with finite increment variance:
\[
\E\sup_{n\ge1}\frac{|M(n)|}{n}<\infty,\qquad \E\sup_{n\ge0}(M(n)-an)^+<\infty\quad(a>0).
\] The first follows by applying Doob's $L^2$ inequality on dyadic intervals. For the second, Lindley's recursion $W_{n+1}=(W_n+X_{n+1})^+$ for increments of mean $-a$ has the law of the running maximum up to time $n$; its expectations are nondecreasing, while $\E(W_{n+1}^2-W_n^2)\le-2a\E W_n+\E X_1^2$, which bounds them uniformly.

We show that the additional waiting time beyond $\tau_*$ has bounded expectation after averaging over records. For $j\ne i,j^*$, put $a_j:=(\rho d_j-g_j)/2>0$ and define the centered random walks $D_j(n):=S^F_{ij}(n)-(d_j/d_{j^*})S^F_{ij^*}(n)$ and $U_j(n):=M^G_{ij}(n)-(d_j/d_{j^*})M^G_{ij^*}(n)$. Since $S^F_{ij^*}(\tau_*)\ge c_{j^*}$, the residual requirement $r_j:=[c_j-S^F_{ij}(\tau_*)]^+$ satisfies
\begin{equation}\label{eq:unique-residual}
r_j\le[U_j(T)-a_jT]^++[-D_j(\tau_*)-a_jT]^++\lambda|1-d_j/d_{j^*}|.
\end{equation}
The first term has uniformly bounded $G_i$-expectation by the negative-drift maximum bound. For the second, $(x-a)^+\le x^2/(4a)$ and the stopped-martingale identity give $\E_i^F[(-D_j(\tau_*)-a_jT)^+\mid R]\le C_{F,G}(b+1)/T$. Its average over $G_i$ is bounded uniformly over stopping policies: since $T\ge1$, $b/T\le g_{j^*}+\lambda+\sup_{n\ge1}|M^G_{ij^*}(n)|/n$. Thus $\E_i^G\E_i^F r_j\le C_{F,G}$ for every $j\ne i,j^*$.

At time $\tau_*$ the primary boundary has already been crossed. Conditional on $R$ and the observations through $\tau_*$, the remaining $F$-increments are fresh; hence Lemma~\ref{lem:passage}, applied to the residual levels, bounds the expected additional waiting time by $v+C_F(1+\sqrt v)$, where $v:=\max_{j\ne i}r_j/d_j$ (and $r_{j^*}:=0$). The preceding bounds and Jensen's inequality therefore give
\begin{equation}\label{eq:unique-wait}
\E_i^G\E_i^F[\tau^+(i,R)-\tau_*\mid R]\le C_{F,G}.
\end{equation}
This estimate allows already crossed boundaries to fall below their levels again: the passage lemma requires all boundaries to hold simultaneously at the final time.

It remains to preserve the leading coefficient in the cost of a correctly proposed state. Write $A:=\{\tau^+(i,R)<\tau^-(i)\}$ and $\sigma:=\sigma(i,R)$. On $A$, the primary evidence at $\sigma$ is at most $b$ plus its overshoot at $\tau_*$ and the sum of absolute primary increments between $\tau_*$ and $\tau^+$. On $A^c$, $\sigma$ equals $\tau_{ji}$ for some $j\ne i$, using the notation of \eqref{eq:unsuccessful-time}. The same change of measure as in that equation gives
\[
\E_i^F[|S^F_{ij^*}(\sigma)|\one_{A^c}]\le e^{-\lambda}\sum_{j\ne i}\E_j^F\sum_{t\le\tau_{ji}}|\Delta S^F_{ij^*}(t)|\le C_F.
\]
The last expectation is finite by Wald's identity and Lemma~\ref{lem:passage}. Indeed, $\E_j^F\tau_{ji}<\infty$, while $\E_j^F|\Delta S^F_{ij^*}(1)|<\infty$ by Assumption~\ref{ass:regular}, since $\log(dF_i/dF_{j^*})=\log(dF_i/dF_j)+\log(dF_j/dF_{j^*})$.
Applying Wald's identity to $S^F_{ij^*}(\sigma)$ and to the absolute increments between $\tau_*$ and $\tau^+$ now yields, conditional on each record,
\begin{equation}\label{eq:unique-acceptance-cost}
\E_i^F\sigma(i,R)\le p_i(i,R)\,\frac{b}{d_{j^*}}+C_F\bigl(1+\E_i^F[\tau^+(i,R)-\tau_*\mid R]\bigr).
\end{equation}
The factor $p_i(i,R)$ is essential. Using $w_i(R)p_i(i,R)\le1$ from \eqref{eq:a-bounds}, the bound \eqref{eq:other-duration} for incorrect candidate states, and the restarting identity \eqref{eq:H-cost}, we obtain from \eqref{eq:unique-wait}
\[
\E_i^F N\le\frac{\E_i^G[(S^G_{ij^*}(T)+\lambda)^+]}{d_{j^*}}+C_{F,G}.
\]
Finally, Wald's identity and the negative-drift maximum bound give $\E_i^G[(S^G_{ij^*}(T)+\lambda)^+]\le g_{j^*}\E_i^G T+\lambda+\E_i^G\sup_{n\ge0}[-S^G_{ij^*}(n)]^+\le g_{j^*}\E_i^G T+C_G$. This proves \eqref{eq:unique-conversion} when $T\ge1$. If $T=0$ is possible, condition on the state-independent initial randomization as in the proof of Theorem~\ref{thm:conversion}; multiplying the conditional bound by the probability of sampling preserves the same constant. Since the state space is finite, one constant works for all states.
\end{proof}

\section{Alternative Orders}\label{app:alternative-orders}

\subsection{Alternative Thresholds}\label{app:thresholds}

We first consider a version of large-budget dominance in which
the budget threshold may depend on the decision problem.

\begin{definition}[Problem-by-problem large-budget dominance]
Experiment $F$ \emph{problem-by-problem large-budget dominates}
experiment $G$ if, for every finite decision problem $\mathcal D$,
there exists $\bar n_{\mathcal D}<\infty$ such that
$V_n(F;\mathcal D)\ge V_n(G;\mathcal D)$ for every
$n\ge\bar n_{\mathcal D}$.
\end{definition}

We next consider a version of small-cost dominance with a common
cost threshold across decision problems with comparable stakes.
Recall that $v(\mathcal D)$ is the value of information and
$\ell^*(\mathcal D)$ is the maximum statewise loss on the support
of the prior. For $\underline v>0$ and $\bar\ell<\infty$, let
\[
\mathbb D(\underline v,\bar\ell)
:=\{\mathcal D:v(\mathcal D)\ge\underline v,\;
\ell^*(\mathcal D)\le\bar\ell\}.
\]

\begin{definition}[Uniform small-cost dominance]
Experiment $F$ \emph{uniformly small-cost dominates}
experiment $G$ if, for every $\underline v>0$ and
$\bar\ell<\infty$, there exists $\bar c>0$ such that
$W_c(F;\mathcal D)\ge W_c(G;\mathcal D)$ for every
$\mathcal D\in\mathbb D(\underline v,\bar\ell)$
and every $c\in(0,\bar c)$.
\end{definition}

\begin{proposition}\label{prop:thresholds}
Let $F$ and $G$ be non-tied experiments. The following
statements are equivalent:
\begin{enumerate}[label=\textup{(\roman*)}]
\item $F$ problem-by-problem large-budget dominates $G$;
\item $F$ uniformly small-cost dominates $G$;
\item $I_{ij}(F)>I_{ij}(G)$ for every ordered pair
of distinct states $i$ and $j$.
\end{enumerate}
\end{proposition}

\begin{proof}
Suppose (iii) holds. By Theorem~\ref{thm:generic-equivalence},
$F$ large-budget dominates $G$ with a threshold common to
all decision problems, which immediately gives (i).
For any $\underline v>0$ and $\bar\ell<\infty$,
Lemma~\ref{lem:budget-cost} then gives a cost threshold
that works for every $\mathcal D\in
\mathbb D(\underline v,\bar\ell)$, proving (ii).

Conversely, suppose (i) or (ii) holds. We show that $F$
small-cost dominates $G$. Fix a decision problem $\mathcal D$.
If $v(\mathcal D)=0$, stopping immediately is optimal under
either experiment. Otherwise, under (i), apply
Lemma~\ref{lem:budget-cost} with the problem-dependent
budget threshold, $\underline v=v(\mathcal D)$, and
$\bar\ell=\ell^*(\mathcal D)$. Under (ii), the same
small-cost comparison follows directly because
$\mathcal D\in\mathbb D(v(\mathcal D),\ell^*(\mathcal D))$.
Thus either condition implies small-cost dominance,
which gives (iii) by Theorem~\ref{thm:generic-equivalence}
and non-tiedness.
\end{proof}

\subsection{Discounting}\label{app:discounting}

We now allow both terminal payoffs and sampling costs to be
discounted. For a cost per observation $c>0$ and a discount
rate $r\ge0$, define
\[
W_{c,r}(E;\mathcal D)
:=\sup_{(T,a)}\sum_i\mu_i\E_i^E
\left[e^{-rT}u(a,i)
-c\sum_{t=1}^T e^{-r(t-1)}\right].
\]
We restrict attention to decision problems with nonnegative
payoffs.

\begin{definition}[Discounted small-cost dominance]
Experiment $F$ \emph{discounted small-cost dominates}
experiment $G$ if, for every finite decision problem
$\mathcal D$ with nonnegative payoffs, there exist
$\bar c_{\mathcal D},\bar r_{\mathcal D}>0$ such that
$W_{c,r}(F;\mathcal D)\ge W_{c,r}(G;\mathcal D)$
whenever $0<c<\bar c_{\mathcal D}$ and
$0\le r<\bar r_{\mathcal D}$.
\end{definition}

\begin{proposition}\label{prop:discounting}
Let $F$ and $G$ be non-tied experiments. Then $F$
discounted small-cost dominates $G$ if and only if
$I_{ij}(F)>I_{ij}(G)$ for every ordered pair
of distinct states $i$ and $j$.
\end{proposition}

\begin{proof}
For necessity, take $r=0$ and consider the binary
decision problems with zero-one payoffs used in the
proof of Theorem~\ref{thm:generic-equivalence}.
Discounted small-cost dominance implies
$I_{ij}(F)\ge I_{ij}(G)$ for every $i\ne j$,
and non-tiedness makes these inequalities strict.

For sufficiency, suppose $I_{ij}(F)>I_{ij}(G)$
for every $i\ne j$, so that
$\rho_i:=\rho_i(F\to G)<1$ in every state.
Fix a finite decision problem
$\mathcal D=(\mu,\mathcal A,u)$ with
nonnegative payoffs.

Let $S:=\supp\mu$, $\bar u_i:=\max_au(a,i)$, and $V^*:=\sum_i\mu_i\bar u_i$. If one action is optimal in every state in $S$, taking it at time zero attains $V^*$ with either experiment. Otherwise,
\begin{equation}\label{eq:equiv-no-common}
 \bigcap_{i\in S}\argmax_{a}u(a,i)=\varnothing.
\end{equation}
Assume henceforth that \eqref{eq:equiv-no-common} holds. Put $\varepsilon:=c+r$ and $\varphi:=\log(1/\varepsilon)$. Constants below may depend on $F$, $G$, and $\mathcal D$, but not on how $\varepsilon$ is split between $c$ and $r$.

For a stopping time $T$, let $T_r:=\sum_{t=1}^Te^{-r(t-1)}$ be its discounted length, so that $1-e^{-rT}=(1-e^{-r})T_r$, and let $c_i:=c+(1-e^{-r})\bar u_i$. The loss of a policy $(T,a)$ in state $i$ relative to $\bar u_i$ is then
\begin{equation}\label{eq:equiv-exact-loss}
 L_i(T,a)=\E_i\bigl[e^{-rT}\ell(a,i)\bigr]+c_i\,\E_iT_r.
\end{equation}
Both terms are nonnegative. The coefficient $c_i$ is the cost of an observation plus the loss from postponing the payoff $\bar u_i$ by one period; it equals $c+r\bar u_i+O(r^2)$, but no approximation is needed in \eqref{eq:equiv-exact-loss}.

\paragraph{A nearly optimal policy can be capped.} Let $\hat\theta_m$ be the maximum-likelihood estimate of the state after $m$ observations of $G$. By Lemma~\ref{lem:ml}, $\Prob_i^G(\hat\theta_m\ne i)\le(K-1)\eta_G^m$. Fix $\ell_0$ large enough, and let $m:=\lceil\ell_0\varphi\rceil$, so that this probability is $O(\varepsilon^3)$ in every state. Taking exactly $m$ observations and then the action that is optimal in state $\hat\theta_m$ has loss $O(\varepsilon\varphi)$ by \eqref{eq:equiv-exact-loss}, because $c_i=O(\varepsilon)$ and $m_r\le m$.

Start with an $\varepsilon^2$-optimal policy for $G$ and cap it at $m$: on the event that it has not stopped by time $m$, stop and take an action that is optimal in state $\hat\theta_m$. When the estimate is correct, this change weakly increases the realized payoff, because the new action is optimal in the state, its payoff is received earlier, and no further costs are incurred; here we use that payoffs are nonnegative. When the estimate is wrong, the payoff falls by at most $\max_{i,a}u(a,i)$. Since the probability of an incorrect estimate is $O(\varepsilon^3)$ in every state, the capped policy, still denoted $(T,a)$, is $O(\varepsilon^2)$-optimal, satisfies $T\le m$, and has total loss $O(\varepsilon\varphi)$. Let
\[
 t_i:=\E_i^GT,\qquad q_i:=\E_i^G\,\ell(a,i),
\]
so that $q_i$ is the undiscounted expected loss from the action in state $i$. Since every state in $S$ has positive prior probability, the terms in \eqref{eq:equiv-exact-loss} are nonnegative, $rm\to0$, and $L_i(T,a)\ge e^{-rm}q_i$,
\begin{equation}\label{eq:equiv-regret}
 q_i=O(\varepsilon\varphi)\qquad(i\in S).
\end{equation}

\paragraph{Some valuable state requires many observations.} Condition \eqref{eq:equiv-no-common} and the nonnegativity of payoffs imply that some fixed state $i_*\in S$ has $\bar u_{i_*}>0$. By \eqref{eq:equiv-regret}, the probability of an action that is optimal in state $i_*$ tends to one in that state, so some such action $a_*$ has probability at least $1/(2|\mathcal A|)$ in state $i_*$. By \eqref{eq:equiv-no-common}, $a_*$ is not optimal in some state $j\in S$, where its probability is at most $q_j/\ell(a_*,j)=O(\varepsilon\varphi)$. The action $a_*$ and the state $j$ may depend on $(c,r)$; the constants are uniform because only finitely many positive losses and KL divergences are involved. Lemma~\ref{lem:likelihood}, applied to the indicator of $\{a=a_*\}$, and \eqref{eq:bernoulli-bound} give
\[
 I_{i_*j}(G)\,t_{i_*}\ge\frac{\varphi-\log\varphi-O(1)}{2|\mathcal A|}-\log2.
\]
Thus, for some $b>0$ and all small enough $\varepsilon$,
\begin{equation}\label{eq:equiv-effort}
 t_{i_*}\ge b\varphi.
\end{equation}

\paragraph{Reproduction improves the value.} Theorem~\ref{thm:conversion} yields a policy $(N,a')$ for $F$ that generates the same distribution of actions as $(T,a)$, so that $\E_i^F\ell(a',i)=q_i$, and satisfies $\E_i^FN\le\rho_it_i+A_{F,G}\sqrt{t_i}$. Since $e^{-rN}\le1$ and $N_r\le N$, its loss satisfies
\[
 L_i^F\le q_i+c_i\bigl(\rho_it_i+A_{F,G}\sqrt{t_i}\bigr).
\]
For the capped policy for $G$, $T_r\ge e^{-rm}T$ gives $L_i^G\ge e^{-rm}(q_i+c_it_i)$. Subtracting,
\begin{equation}\label{eq:equiv-gain}
 L_i^G-L_i^F\ge c_i\bigl[(e^{-rm}-\rho_i)\,t_i-A_{F,G}\sqrt{t_i}\bigr]-(1-e^{-rm})\,q_i.
\end{equation}
Let $\delta:=1-\max_i\rho_i>0$. Since $rm=O(\varepsilon\varphi)=o(1)$, we have $e^{-rm}-\rho_i\ge\delta/2$ for all sufficiently small $\varepsilon$. The function $t\mapsto(\delta/2)t-A_{F,G}\sqrt t$ is bounded below, so the first term in \eqref{eq:equiv-gain} is at least $-O(\varepsilon)$ in every state. In state $i_*$, however, \eqref{eq:equiv-effort} and $c_{i_*}\ge\min\{1,\bar u_{i_*}/2\}\varepsilon$ make this term at least a positive constant times $\varepsilon\varphi$. Moreover, by \eqref{eq:equiv-regret}, the last term in \eqref{eq:equiv-gain} is $O(\varepsilon^2\varphi^2)=o(\varepsilon\varphi)$ in every state with positive prior probability. Since $\mu_{i_*}>0$, the expected payoff gain from reproduction is therefore bounded below by a positive constant times $\varepsilon\varphi$ for all sufficiently small $\varepsilon$. This exceeds the $O(\varepsilon^2)$ optimality gap of the capped policy, proving $W_{c,r}(F;\mathcal D)\ge W_{c,r}(G;\mathcal D)$ uniformly over all $c>0$ and $r\ge0$ with $c+r=\varepsilon$ sufficiently small. Choosing positive thresholds whose sum is below this bound completes the proof.
\end{proof}

\end{document}